\documentclass[12pt,a4paper]{article}

\usepackage{amscd,amsfonts,amsmath,amssymb,amsthm,bbm,bm,latexsym,mathrsfs}
\usepackage{epsfig,graphics,graphicx,natbib,subfigure,longtable}
\usepackage[english]{babel}
\usepackage[usenames]{color}
\usepackage{bookmark}
\usepackage{booktabs}
\usepackage{sgame}
\usepackage[top=1.7in, bottom=1.7in, left=0.9in, right=0.9in]{geometry}
\usepackage{epstopdf}
\allowdisplaybreaks
\makeatother

\newtheorem{theorem}{Theorem}[section]

\theoremstyle{remark}

\theoremstyle{plain}

\begin{document}

%\markboth{F. Leisen, L. Rossini \and C. Villa}{Objective Bayesian Analysis of the Yule--Simon Distribution with Applications}

\title{\vspace{-50pt} Objective Bayesian Analysis of the Yule--Simon Distribution with Applications}

\author{
\hspace{-20pt} 
Fabrizio Leisen\textsuperscript{1} \hspace{15pt}
 Luca Rossini\textsuperscript{2} \hspace{15pt} Cristiano Villa\textsuperscript{1}\footnote{Corresponding Authors: Fabrizio Leisen (\href{fabrizio.leisen@gmail.com}{fabrizio.leisen@gmail.com}); Luca Rossini (\href{luca.rossini@unive.it}{luca.rossini@unive.it}); Cristiano Villa (\href{cv88@kent.ac.uk}{cv88@kent.ac.uk})}
        \\ 
        \vspace{5pt}
        \\
        {\centering {\small \textsuperscript{1} 
        University of Kent, U.K. \hspace{18pt}
        \textsuperscript{2}
            Ca' Foscari University of Venice, Italy}} \vspace{5pt} \\
     }

\date{\today}
\maketitle

%%%%%%%%%%%%%%%%%%%%%%%%%%%%%%%%%%%%%%%%%%%%%%%%%%%%%%
%%         	Abstract			%%%
%%%%%%%%%%%%%%%%%%%%%%%%%%%%%%%%%%%%%%%%%%%%%%%%%%%%%%
\abstract{ The Yule--Simon distribution is usually employed in the analysis of frequency data. As the Bayesian literature, so far, ignored this distribution, here we show the derivation of two objective priors for the parameter of the Yule--Simon distribution. In particular, we discuss the Jeffreys prior and a loss-based prior, which has recently appeared in the literature. We illustrate the performance of the derived priors through a simulation study and the analysis of real datasets.

}

%%In this paper we propose a Bayesian approach for making inference on the shape parameter of the so called Yule--Simon distribution. In particular, we explore methods to define prior distributions for the parameter where prior information is minimal or absent.
%%%%%%%%%%%%%%%%%%%%%%%%%%%%%%%%%%%%%%%%%%%%%%%%%%%%%%
%%    	     	Intro			%%%
%%%%%%%%%%%%%%%%%%%%%%%%%%%%%%%%%%%%%%%%%%%%%%%%%%%%%%
\section{Introduction}
\label{Intro}
In this work we aim to fill a gap in the Bayesian literature by proposing two objective priors for the parameter of the Yule--Simon distribution. The distribution was firstly discussed in \cite{Yule25} and then re-proposed in \cite{Simon55}, and can be used in scenarios where the center of interest is some sort of frequency in the data. For example, \cite{Yule25} used it to model abundance of biological genera, while \cite{Simon55} exploited the distribution properties to model the addition of new words to a text. It goes without saying that other areas of applications can be considered where, for instance, frequencies represent the elementary unit of observation. For example, in this paper we show the employment of the Yule--Simon distribution in modelling daily increments of social network stock options, surnames and 'superstar' success in the music industry.

Despite the wide range of applications, the literature on the Yule--Simon distribution appears to be limited. And, more surprisingly, to the best of our knowledge it seems that no attention has been given to the problem by the Bayesian community. Given the challenges that classical inference faces in estimating the parameter of the distribution \citep{Garcia11}, the possibility of tackling the problem from a Bayesian perspective is, undoubtedly, appealing.

In addressing the estimation of the shape parameter of the Yule--Simon distribution by means of the Bayesian framework, we opted for an objective approach. We propose two priors: the first is the Jeffreys rule prior \citep{JE61}, while the second is obtained by applying the loss-based approach discussed in \cite{VillaWalker15}. Although we formally introduce the Yule--Simon distribution and its derivation in the next Section, it is important to give an anticipation of the general idea here, so to fully appreciate the gain in adopting an objective approach. As nicely illustrated in \cite{ChungCox94}, the shape parameter of the distribution is linked via a one-to-one transformation to the probability that the next observation will not take a value  previously observed. For example, if we have observed $n$ words in a text, we wish to make inference on the probability that the $(n+1)$ observation is a word not yet encountered in the text, assuming this probability to be constant. It is then clear that the Yule--Simon distribution models extremely large events. As such, the information in the data about these events is limited and a ``wrongly'' elicited prior could end up dominating the data. On the other hand, a prior with minimal information content would allow the data ``to speak'', resulting in a more robust inferential procedure. We do not advocate that in every circumstance an objective approach is the only suitable. In fact, if reliable prior information is available, an elicited prior would represent, in general, the natural choice. Alas, in the presence of phenomena with extremely rare events, the above information is often insufficient or incomplete, and an objective choice would then represent the most sensible one.\\

The paper is organized as follows. In Section \ref{Prelim} we set the scene by introducing the Yule--Simon distribution and the notation that will be used throughout the paper. The proposed objective priors are derived and discussed in Section \ref{ObjectivePriors}. Section \ref{Simu} collects the analysis of the frequentist performances of the posterior distributions yielded by the proposed priors. Through a set of several simulation scenarios, we compare and analyse the inferential capacity of the objective priors here discussed. In Section \ref{RealData} we illustrate the application of the priors to three real-data applications. Finally, Section \ref{Concl} is reserved to concluding remarks and points of discussion.

%%%%%%%%%%%%%%%%%%%%%%%%%%%%%%%%%%%%%%%%%%%%%%%%%%%%%%
%% 	 	Preliminaries		%%%
%%%%%%%%%%%%%%%%%%%%%%%%%%%%%%%%%%%%%%%%%%%%%%%%%%%%%%
\section{Preliminaries}
\label{Prelim}
The most known functional form of the Yule--Simon distribution, possibly, is the following:
\begin{equation}
f(k;\rho) = \rho\, \mbox{B}(k,\rho+1), \qquad k=1,2,\ldots \mbox{ and } \rho>0, \label{Yule_Org}
\end{equation}
where $\mbox{B}(\cdot,\cdot)$ is the beta function and $\rho$ is the shape parameter. The distribution in \eqref{Yule_Org} was firstly proposed by \cite{Yule25} in the field of biology; in particular, to represent the distribution of species among genera in some higher taxon of biotic organisms.
More recently, \cite{Simon55} noticed that the above distribution can be observed in other phenomena, which appear to have no connection among each others. These include, the distribution of word frequencies in texts, the distribution of authors by number of scientific articles published, the distribution of cities by population and the distribution of incomes by size. The derivation process followed by \cite{Yule25} was based on word frequencies, and it consisted of two assumptions:
\begin{itemize}
	\item[(i)] The probability that the $(n+1)$-th word is a word observed $k$ times in the first $n$ words, is proportional to $k$; and
	\item[(ii)] The probability that the $(n+1)$-th word is new (i.e. not being observed in the first $n$ words) is constant and equal to $\alpha\in(0,1)$.
\end{itemize}
\cite{Yule25} shows that, under the condition of stationarity,  the process defined by the above two assumptions yields \eqref{Yule_Org} by setting $\rho = 1/(1-\alpha)$, obtaining:
\begin{equation}
f(k;\alpha) = \frac{1}{1-\alpha}\mbox{B}\left(k,\frac{1}{1-\alpha}+1\right). \label{Disc}
\end{equation}
An important consequence of the above assumption (ii) is that the shape parameter $\rho$ of the distribution takes values in $(1,+\infty)$. In other words, should we use the model as in \cite{Yule25}, which includes the possibility that $0<\rho\leq1$, we would loose the interpretation of the generating process described by the two assumptions above. In fact, for $\rho<1$, the probability of observing a new word would be negative; while for $\rho=1$ the probability would be zero, rendering the process trivial (i.e. all the observed words will be equal to the first one observed). Furthermore, the expectation of the Yule--Simon distribution is defined only for values of the shape parameter larger than one, and this property is something one would expect in most applications. For all the above reasons, in this work we focus on the parametrization of the Yule-Simon given in \eqref{Disc}, that is  we will discuss prior distributions for $\alpha$.

In addition to the parametrization of the Yule--Simon distribution as in \eqref{Disc}, we will also consider the possibility of having the parameter $\alpha$ discrete. This is a common finding in literature, especially when implementations of the model are considered. See, for example, \cite{Simon55} and \cite{Garcia11}. The discretization of $\alpha$ will be discussed in detail in Section \ref{DiscrOb}.

%The discretized parameter space mainly considered is $\alpha=\{0.1,0.2,\ldots,0.9\}$, which appear to be sufficiently dense to discern adjoining distributions with reasonable sample sizes. However, if the sample size is relatively large, more dense parameter spaces can be considered, such as $\alpha=\{0.01,0.02,\ldots,0.99\}$. In Section \ref{Simu}, we show how a discrete prior defined on different discretizations of $\alpha$ performs.

%%%%%%%%%%%%%%%%%%%%%%%%%%%%%%%%%%%%%%%%%%%%%%%%%%%%%%
%% 	 	Objective			%%%
%%%%%%%%%%%%%%%%%%%%%%%%%%%%%%%%%%%%%%%%%%%%%%%%%%%%%%
\section{Objective Priors for the Yule-Simon distribution}
\label{ObjectivePriors}
This section is devoted to the derivation of two objective priors for the Yule-Simon distribution: the Jeffreys prior and loss-based prior. The former assumes that parameter space of $\alpha$ is continuous and it is based on the well-known invariance property proposed by \cite{JE61}; the latter assumes the parameter space discrete and is based on \cite{VillaWalker15}. %We stress the fact that the reference prior coincides with the Jeffreys prior in the one parameter space \citep{BBS09}; as such, by deriving the Jeffreys prior for the Yule--Walker, we also obtain the reference prior.

%%%%%%%%%%%%%%%%%%%%%%%%%%%%%%%%%%%%%%%%%%%%%%%%%%%%%%
%% 	 	the JeffrePrior		%%%
%%%%%%%%%%%%%%%%%%%%%%%%%%%%%%%%%%%%%%%%%%%%%%%%%%%%%%
\subsubsection*{The Jeffreys Prior}
\label{DiscrJe}
The Jeffreys prior is defined in the following way \citep{JE61}:
\begin{equation}
\pi(\alpha) \propto \sqrt{\mathcal{I}(\alpha)} \notag
\end{equation}
where $\mathcal{I}(\alpha)=\mathbb{E}_{\alpha}\biggl[-\frac{\partial ^{2}\log(f(k;\alpha))}{\partial \alpha^{2}} \biggr]$ is the Fisher Information. In the next Theorem \label{Jef1} (which proof is in the Appendix) an explicit expression of the Jeffreys prior for the Yule-Simon distribution is provided. 
\begin{theorem}\label{Jef1}
Let $f(k;\alpha)$ be the Yule-Simon distribution defined in equation (\ref{Disc}), with $0<\alpha<1$. The Jeffreys prior for $\alpha$ is 
\begin{equation}\label{JefPrior}
\pi(\alpha) \propto q(\alpha) \\
\end{equation}
where 
$$q(\alpha)=\frac{1}{1-\alpha}\sqrt{1-\frac{1}{(2-\alpha)^{2}} \,_{3}F_{2}\biggl(1,\frac{1}{1-\alpha}+1,1; \frac{1}{1-\alpha}+2,\frac{1}{1-\alpha}+2;1\biggr)}.$$
with $_{3}F_{2}$ being the hypergeometric distribution function.
\end{theorem}
\noindent The Jeffreys prior stated in Theorem \ref{Jef1} is a proper prior. In fact, let
\begin{equation}\label{JefPriorNorm}
\pi(\alpha) =\frac{q(\alpha)}{K}, \notag \\
\end{equation}
where
\begin{footnotesize}
\begin{equation}\label{NormConst}
K=\int_0^1\frac{1}{1-\alpha}\sqrt{1-\frac{1}{(2-\alpha)^{2}} \,_{3}F_{2}\biggl(1,\frac{1}{1-\alpha}+1,1; \frac{1}{1-\alpha}+2,\frac{1}{1-\alpha}+2;1\biggr)}d\alpha\notag 
\end{equation}
\end{footnotesize}
is the normalizing constant of $\pi(\alpha)$.
\noindent It is not difficult to prove that $$K<\infty.$$ 
Indeed,  
 $$K\leq\int_0^1\sqrt{\frac{3-\alpha}{1-\alpha}}\frac{1}{2-\alpha}d\alpha=\frac{1}{3}\pi-\ln(2-\sqrt{3})<\infty. $$
The result above follows from the following inequality
$$_{3}F_{2}\biggl(1,\frac{1}{1-\alpha}+1,1; \frac{1}{1-\alpha}+2,\frac{1}{1-\alpha}+2;1\biggr)\geq 1.$$
The properness of the prior in (\ref{JefPrior}) ensures the properness of the yielded posterior distribution for $\alpha$, as such suitable for inference.

%%%%%%%%%%%%%%%%%%%%%%%%%%%%%%%%%%%%%%%%%%%%%%%%%%%%%%
%% 	 	The KL Prior		%%%
%%%%%%%%%%%%%%%%%%%%%%%%%%%%%%%%%%%%%%%%%%%%%%%%%%%%%%
\subsubsection*{The Loss-based Prior}
\label{DiscrOb}
\cite{VillaWalker15} introduced a method for specifying and objective prior for discrete parameters. The idea is to assign a \emph{worth} to each parameter value by objectively measuring what is lost if the value is removed, and it is the true one. The loss is evaluated by applying the well known result in \cite{Berk66} stating that, if a model is misspecified, the posterior distribution asymptotically accumulates on the model which is the nearest to the true one, in terms of the Kullback--Leibler divergence.

Given that the parameter $\alpha\in(0,1)$ of the Yule--Simon is in principle continuous, the above method can not be applied. However, the boundedness of the interval allows for an easy discretization, directly we can consider the set
$$\mathbb{D}_{M}=\left\lbrace \alpha=\frac{i}{M}: i=1,\dots,M-1\right\rbrace.$$
Therefore, the \emph{worth} of the parameter value $\alpha$ is represented by the Kullback--Leibler divergence $D_{KL}(f(k|\alpha)\|f(k|\alpha^\prime))$, where $\alpha^\prime\neq\alpha$ is the parameter value that minimizes the divergence. To link the \emph{worth} of a parameter value to the prior mass, \cite{VillaWalker15} use the self-information loss function. This particular type of loss function measures the loss in information contained in a probability statement \citep{MerFed98}. As we now have, for each value of $\alpha$, the loss in information measured in two different ways, we simply equate them obtaining the loss-based prior of \cite{VillaWalker15}:
\begin{equation}
\pi(\alpha) \propto \exp{\biggl\{\min_{\alpha' \ne \alpha} D_{KL}(f(k|\alpha)\|f(k|\alpha'))\biggr\}} -1 \qquad\alpha,\alpha^\prime\in\mathbb{D}_{M},\label{OBB}
\end{equation}
where
\begin{align*}
D_{KL}(f(k|\alpha)\|f(k|\alpha'))
&=\log{\left(\frac{1-\alpha'}{1-\alpha}\right)}+\mathbb{E}_{\alpha}\left\{\log{\left[B\left(k;\frac{1}{1-\alpha}+1\right)\right]}\right\}\\
&\hspace{3cm}-\mathbb{E}_{\alpha}\left\{\log{\left[B\left(k;\frac{1}{1-\alpha'}+1\right)\right]}\right\}.
\end{align*}
As the discretized parameter space is finite, no matter what value of $M$ one chooses, the prior \eqref{OBB} is proper, hence, the yielded posterior will be proper as well.

An important aspect is that the value $\alpha^\prime$ minimizing the Kullback--Leibler divergence can not be analytically determined, and the prior has to be computationally derived. However, even for large values of $M$, the computational cost is trifling compared to the whole Monte Carlo procedure necessary to simulate from the posterior distribution. \\

\begin{figure}[h!]
\centering
  {\includegraphics[width=7.5cm]{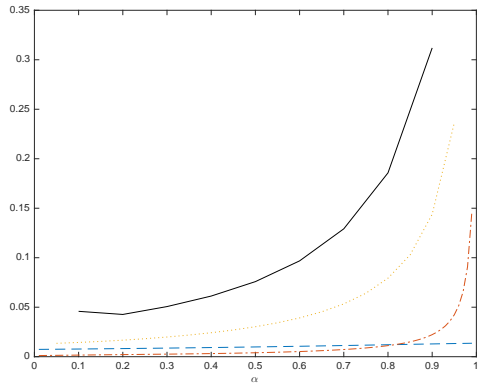}}
  \caption{Prior distribution for $\alpha$ obtained by applying Jeffreys rule (dashed line), the loss-based method with $M=10$ (continuous line), with $M=20$ (dotted line) and with $M=100$ (dash-dotted line).}
\label{JJvsKL}
\end{figure}
To have a feeling of the prior distributions derived above, we have plotted them in Figure \ref{JJvsKL}. The behaviour of the  priors is similar, in the sense that they tend to increase as $\alpha$ increases and, for increasing values of $M$, the two distributions seem to converge. However, we note that the Jeffreys prior is flatter than the loss-based priors for large values of the parameter, i.e. for $\alpha$ approximately greater than $0.8$.

%%%%%%%%%%%%%%%%%%%%%%%%%%%%%%%%%%%%%%%%%%%%%%%%%%%%%%
%% 	 	Simulation 		%%%
%%%%%%%%%%%%%%%%%%%%%%%%%%%%%%%%%%%%%%%%%%%%%%%%%%%%%%
\section{Simulation Study}
\label{Simu}
The objective priors defined in Section \ref{ObjectivePriors} are automatically derived by taking into consideration properties intrinsic to the Yule--Simon distribution. In other words, they do not depend on experts knowledge or previous observations. It is therefore necessary, in order to validate them, to assess the goodness of the priors by making inference on simulated data. This section is dedicated in performing a simulation study on the parameter $\alpha$ using observations obtained from fully known distributions.

We have considered different sample sizes, $n=30$, $n=100$ and $n=500$, to analyse the behaviour of the prior distributions under different level of information coming from the data. Here we show the results for $n=100$ only, as the sole differences in using $n=30$ and $n=500$ sample sizes are limited to the precision of the inferential results: relatively low for $n=30$ and relatively high for $n=500$, as one would expect. Besides that, the differences in the performance of the two priors noted for $n=100$ remain for the other sample sizes. As the loss-based prior depends on the discretization of the parameter space, for illustration purposes, we have considered $M=10$ and $M=20$, that is $\alpha\in\{0.1,0.2,\ldots,0.9\}$ and $\alpha\in\{0.05,0.10,\ldots,0.95\}$, respectively.

Both the Jeffreys prior and the loss-based prior yield posterior distributions for $\alpha$ which are not analytically tractable, hence, it is necessary to use Monte Carlo methods. We have generated $100$ samples from a Yule--Simon distribution with the parameter $\alpha$ set to every value in the parameter space, 9 for $M=10$ and 19 for $M=20$. For each sample we have simulated from the posterior distribution of $\alpha$, under both priors, by running $10,000$ iterations, with a burn-in period of $2,000$ iterations.

To evaluate the priors we have considered two frequentist measures. The first is the frequentist coverage of the $95\%$ credible interval. That is, for each posterior, we compute the interval between the 0.025 and 0.975 quantiles and see if the true value of $\alpha$ is included in it. Over repeated samples, one would expect a proportion of about $95\%$ of the posterior intervals to contain the true parameter value. The second frequentist measure gives an idea of the precision of the inferential process, and it is represented by the square root of the mean squared error (MSE) from the mean, relative to the parameter value: $\sqrt{\mbox{MSE}(\alpha)}/\alpha$. We have considered the MSE from the median as well but, due to the approximate symmetry of the posterior, the results are very similar to the MSE from the mean.
\begin{figure}[h!]
\centering
 %\subfigure[]%[Mean square error from the mean]
  %{\includegraphics[width=6.75cm]{./FiguresNew/PlotMeanSEn100A01.eps}}
 %\subfigure[Mean square error from the median]   
    %{\includegraphics[width=6.75cm]{./FiguresNew/PlotMedianSEn100A01.eps}}
    \subfigure[]%[Frequentist coverage with $95\%$ line]
   {\includegraphics[width=6.75cm]{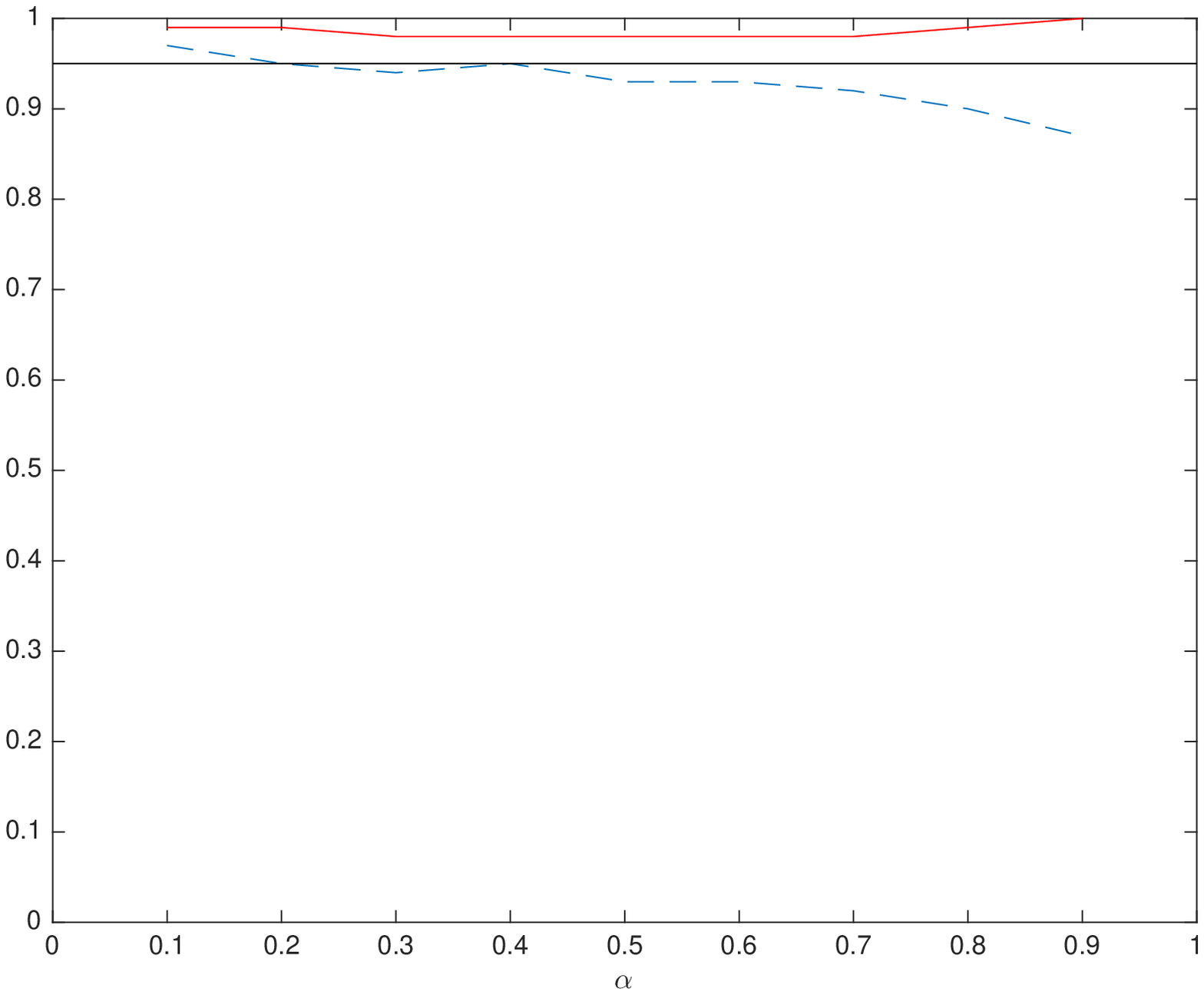}}
	\subfigure[]%[Relative Mean square error from the mean]
 	{\includegraphics[width=6.75cm]{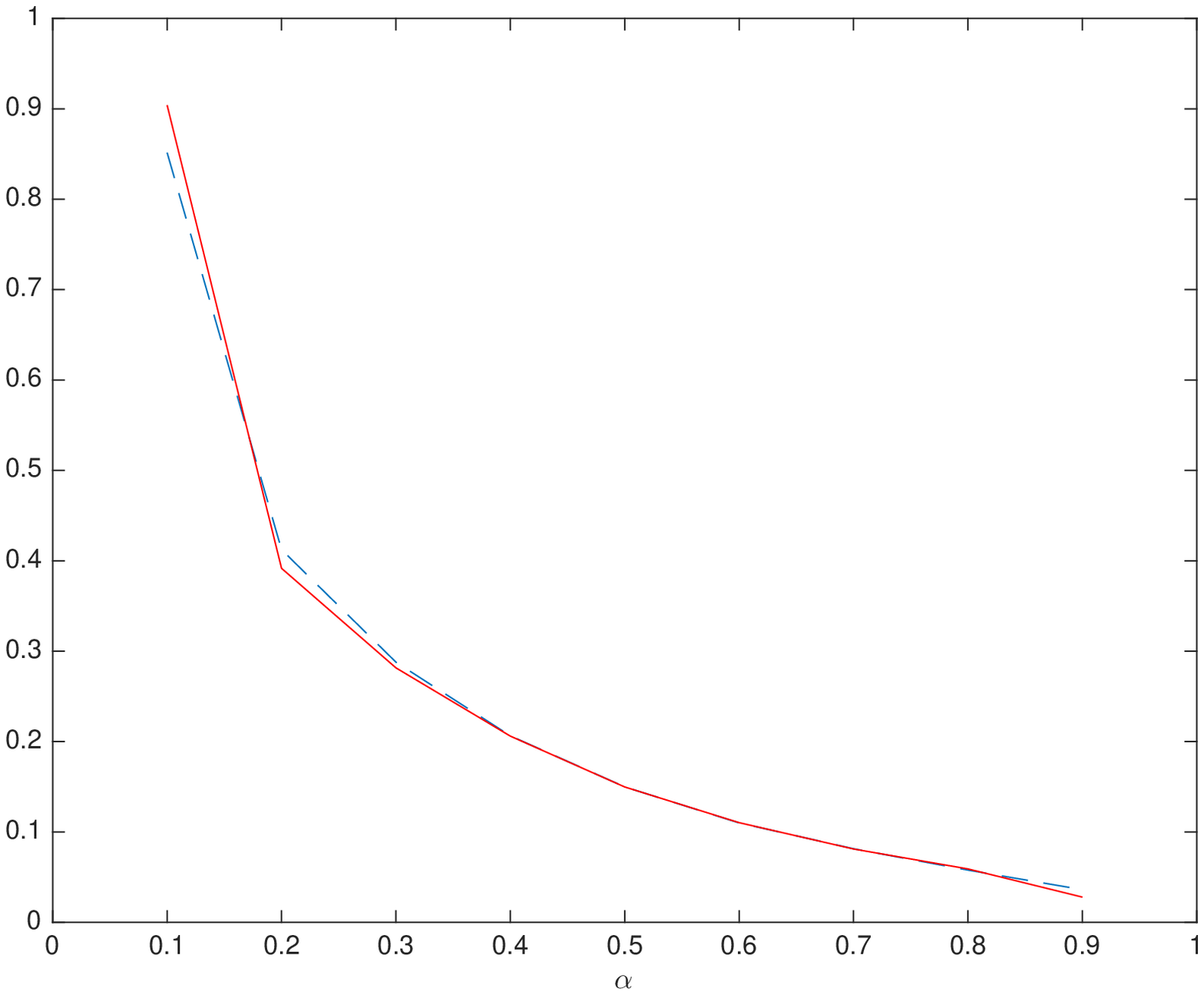}}
  %\subfigure[Relative Mean square error from the median]   
    %{\includegraphics[width=6.75cm]{./FiguresNew/PlotsqrtMedianSe01.eps}} 
       % \subfigure[]%[Frequentist coverage]
   %{\includegraphics[width=6.75cm]{./FiguresNew/PlotCoveragen100A01.eps}}
  \caption{Frequentist properties of the Jeffreys prior (dashed line) and the loss-based prior (continuous line) for $n=100$. The loss-prior is considered on the discretized parameter space with $M=10$. The left plot shows the posterior frequentist coverage of the $95\%$ credible interval, and the right plot represents the square root of the MSE from the mean of the posterior, relative to $\alpha$.}
\label{alpha01}
\end{figure}
Figure \ref{alpha01} details the results for the simulations with $n=100$ and a parameter space for $\alpha$ discretized with increments of $0.1$, that is $\alpha\in\{0.1,0.2,\ldots,0.9\}$. If we compare the coverage, we note that the loss-based prior tends to over-cover the credible interval, while the Jeffreys prior, although shows a better coverage for values of $\alpha<0.5$, deteriorates in performance as the parameter tends to the upper bound of its space. Looking at the MSE, both priors appear to have very similar performance, and the (relative) error tends to decrease and $\alpha$ increases.
\begin{figure}[h!]
\centering
  %\subfigure%[Mean square error from the mean]
  %{\includegraphics[width=6.75cm]{./FiguresNew/PlotMeanSEn100A005.eps}}
 %\subfigure[Mean square error from the median]   
    %{\includegraphics[width=6.75cm]{./FiguresNew/PlotMedianSEn100A005.eps}}
    \subfigure[]%[Frequentist coverage with $95\%$ line]
   {\includegraphics[width=6.75cm]{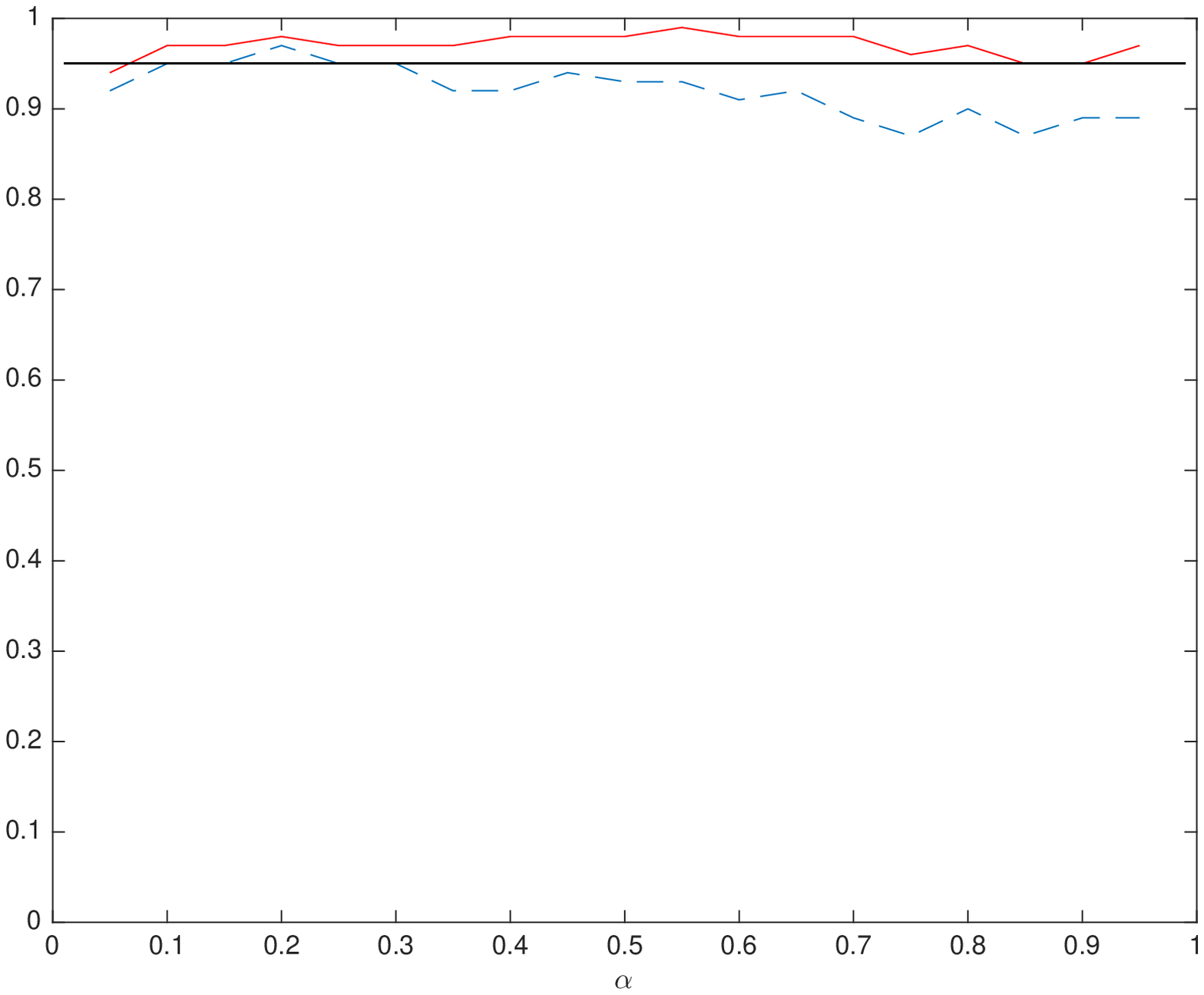}}
	\subfigure[]%[Relative Mean square error from the mean]
	{\includegraphics[width=6.75cm]{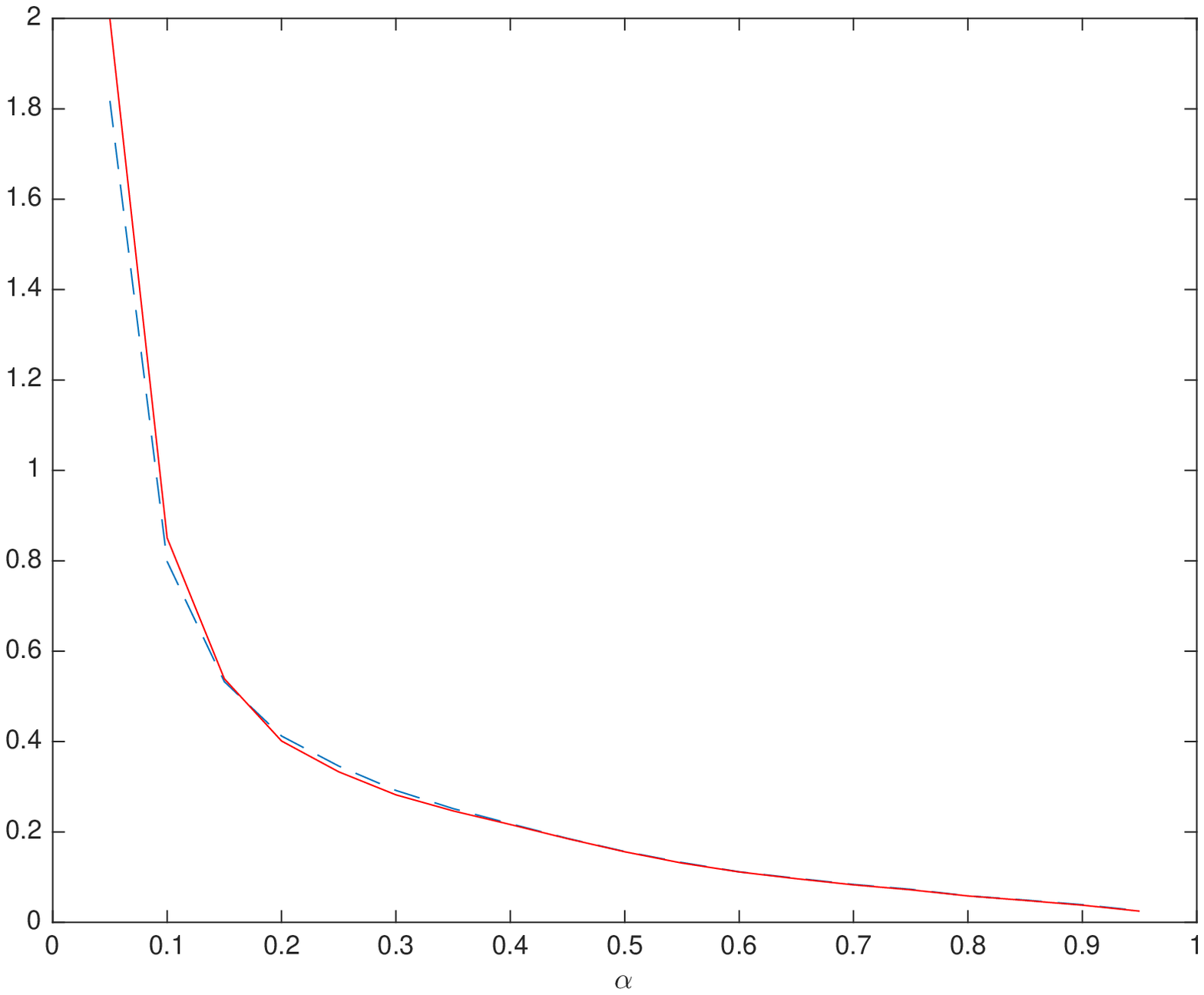}}
  %\subfigure[Relative Mean square error from the median]   
    %{\includegraphics[width=6.75cm]{./FiguresNew/PlotsqrtMedianSe005.eps}} 
       % \subfigure[Frequentist coverage]
   %{\includegraphics[width=6.75cm]{./FiguresNew/PlotCoveragen100A005.eps}}
  \caption{Frequentist properties of the Jeffreys prior (dashed line) and the loss-based prior (continuous line) for $n=100$. The loss-prior is considered on the discretized parameter space with $M=20$. The left plot shows the posterior frequentist coverage of the $95\%$ credible interval, and the right plot represents the square root of the MSE from the mean of the posterior, relative to $\alpha$.}
\label{alpha005}
\end{figure}
In Figure \ref{alpha005} we have compared the frequentist performance of the Jeffreys prior with the loss-based prior defined over a more densely discretized parameter space, i.e. $\alpha=\{0.05,0.10,\ldots,0.95\}$. We note a smoother behaviour of the priors compared to Figure \ref{alpha01}, which is obviously due to the denser characterization considered. The coverage still reveals a tendency of the loss-based prior to over-cover, although less pronounced than the previous case. Jeffreys prior does not present any significant difference from the previous case, as one would expect. For what it concerns the MSE, the differences between the two priors are negligible, and the only aspect we note, as mentioned above, is a smoother decrease of the error as the parameter increases.

%25
We look more into the details of the objective approach by analysing two i.i.d. samples. In particular, we consider a random sample of size $n=100$ from a Yule--Simon distribution with $\alpha=0.40$ and a sample, of the same size, from a Yule--Simon with $\alpha=0.68$.

In both cases, we have sampled from the posterior distribution via Monte Carlo methods with $10,000$ iterations and a burn-in period of $2,000$ iterations.
\begin{figure}[h!]
\centering
	\subfigure[]%[Jeffreys]
	{\includegraphics[width=6.75cm]{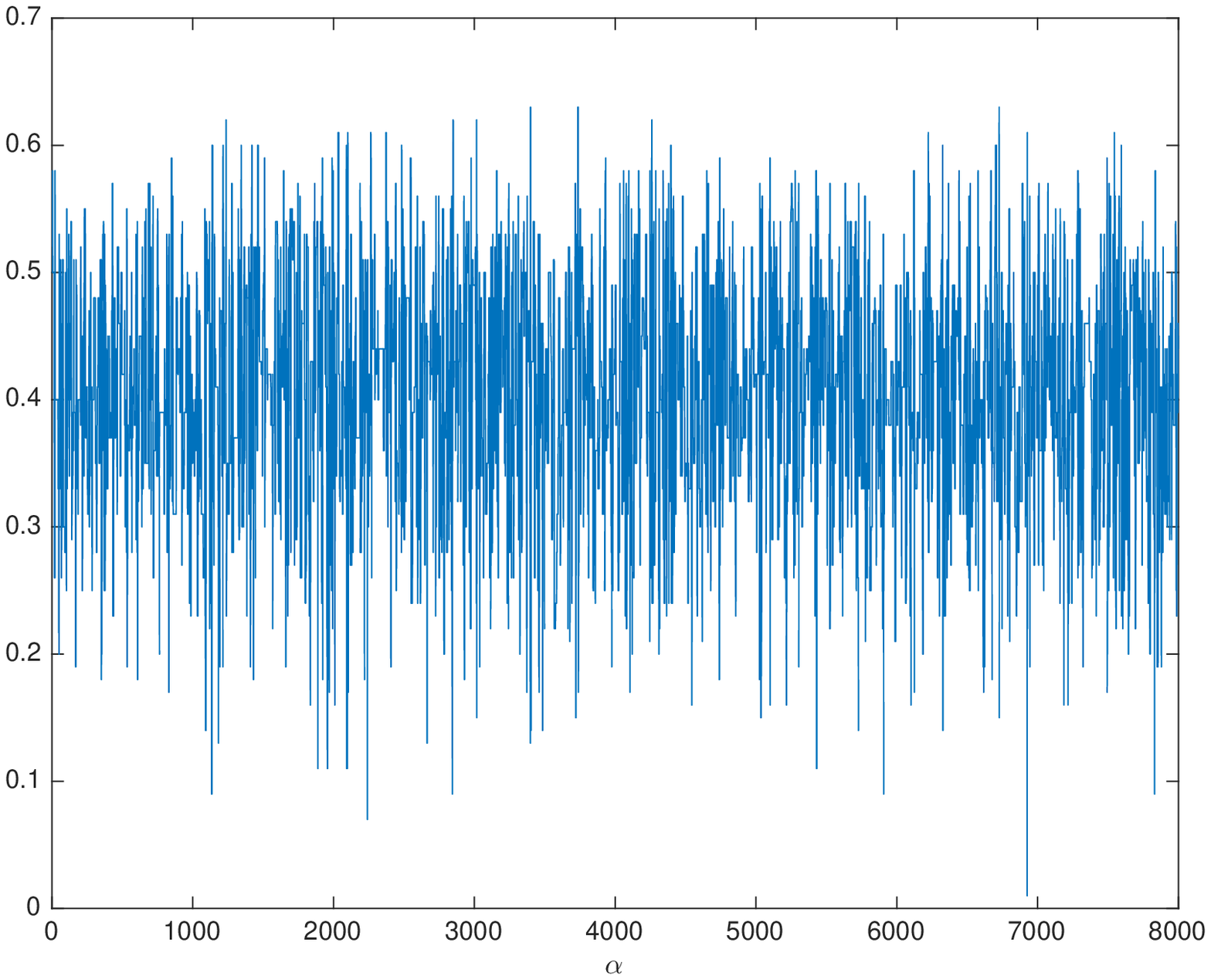}}
	\subfigure[]%[Jeffreys]
	{\includegraphics[width=6.75cm]{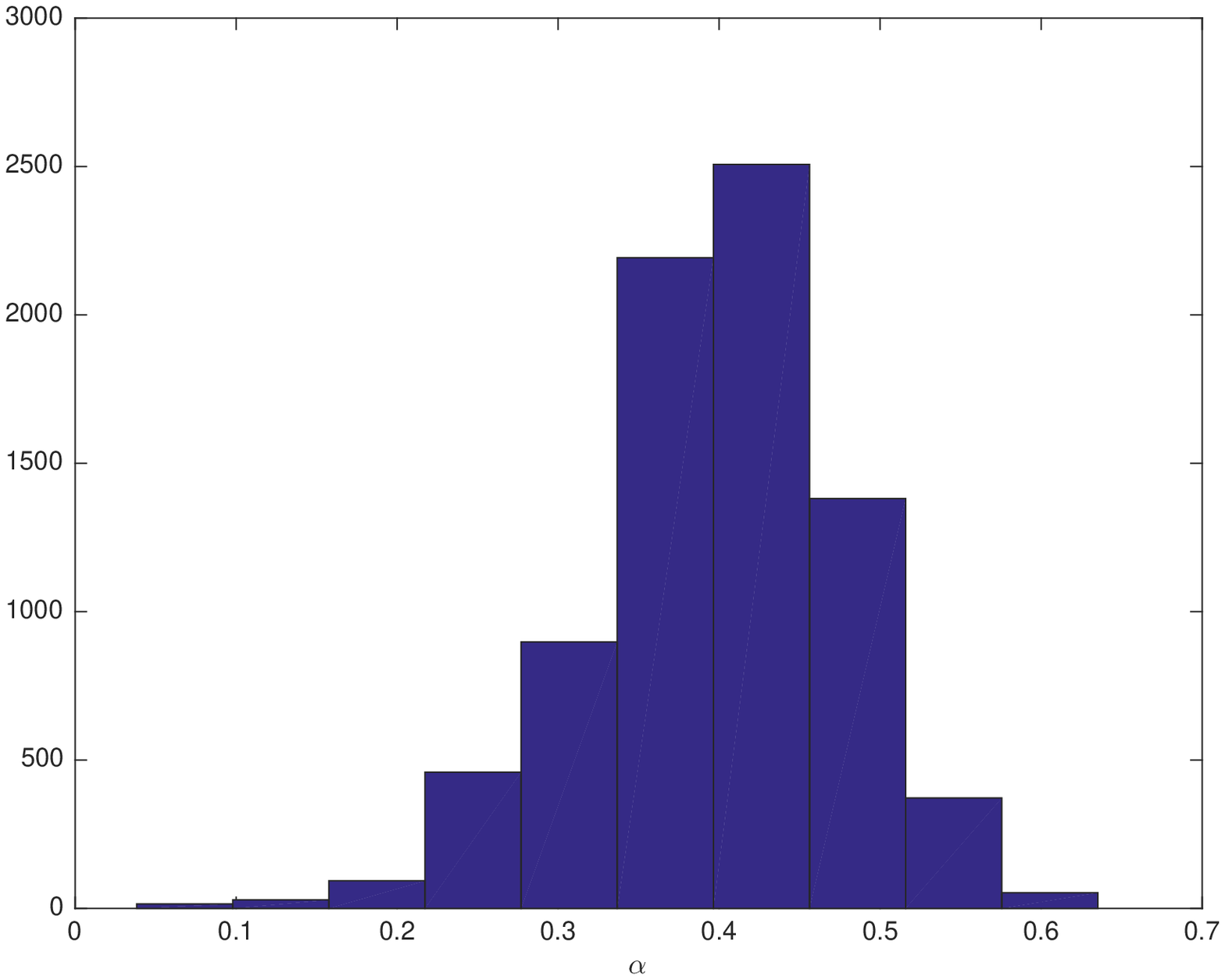}}
	\subfigure[]%[loss-prior with $M=10$]
    {\includegraphics[width=6.75cm]{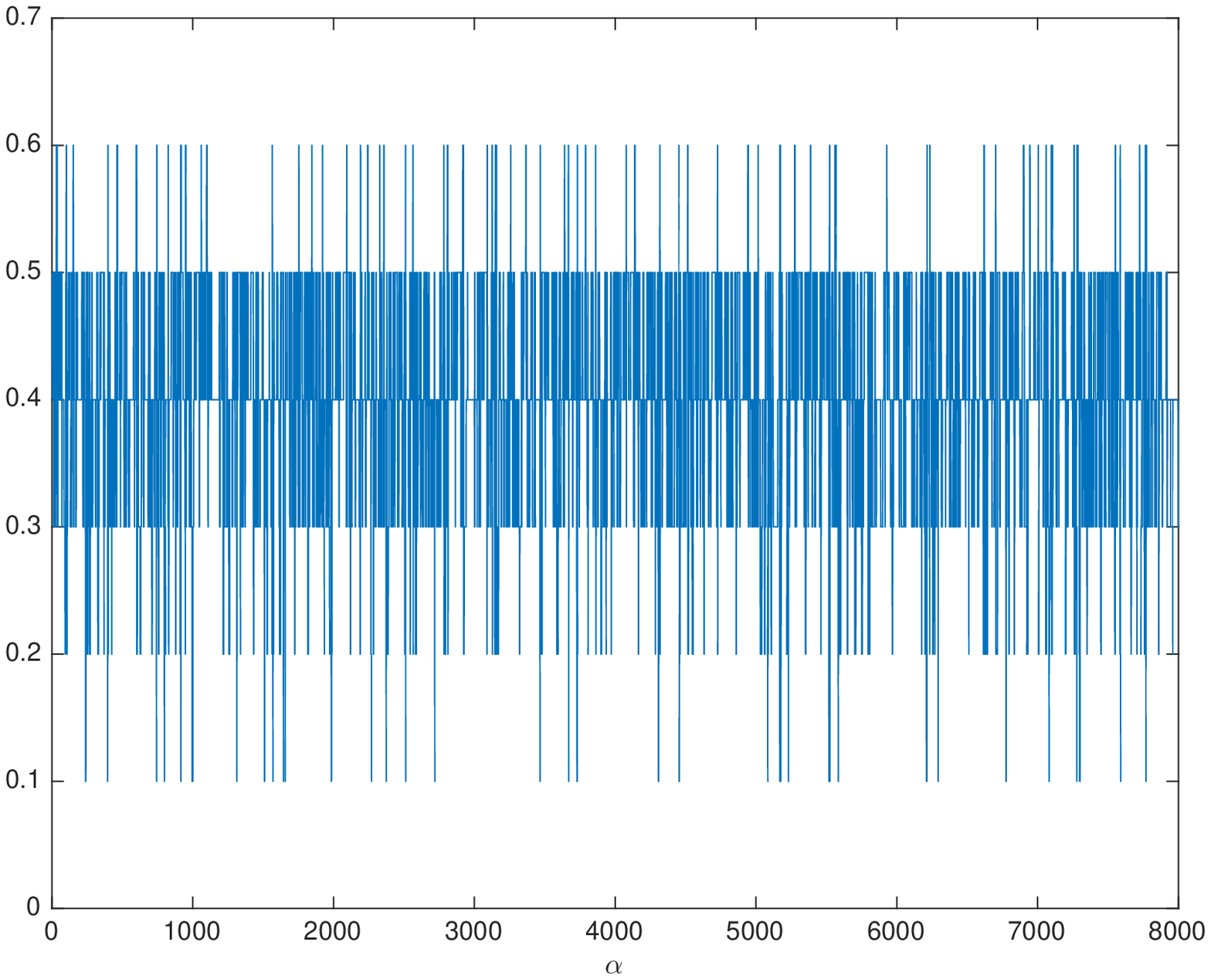}}
    \subfigure[]%[loss-prior with $M=10$]
    {\includegraphics[width=6.75cm]{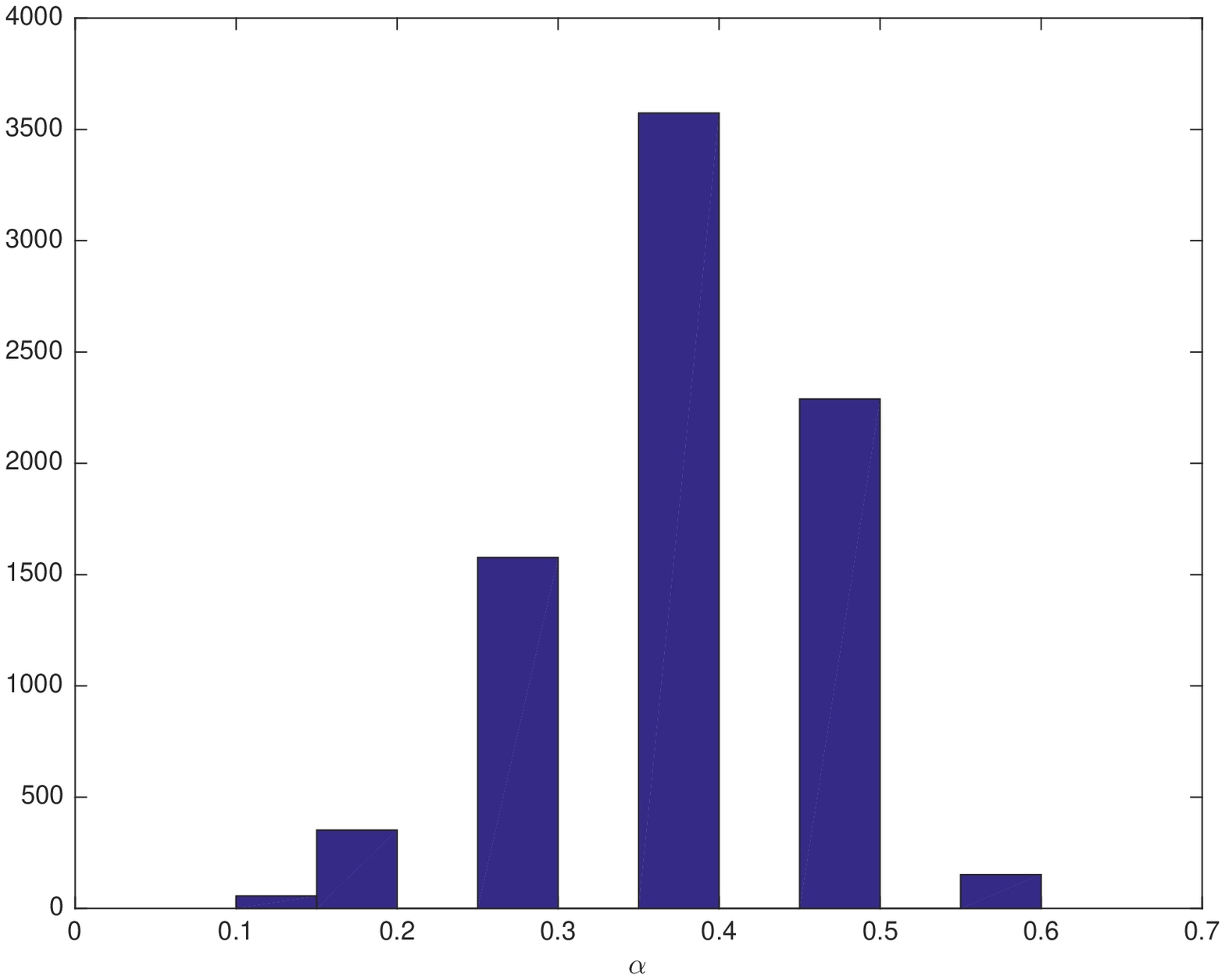}}
    \subfigure[]%[loss-prior with $M=20$]
	{\includegraphics[width=6.75cm]{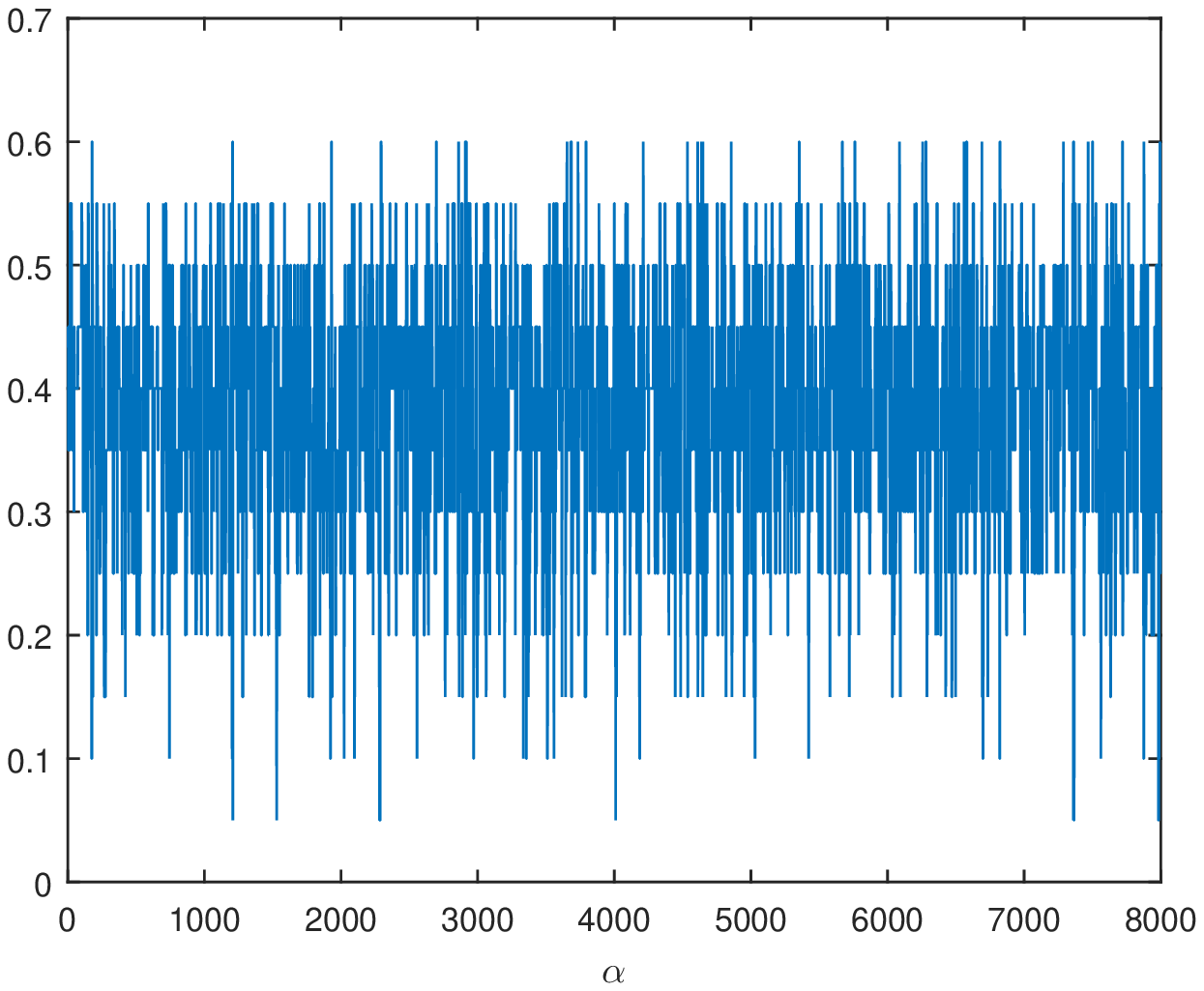}}
    \subfigure[]%[loss-prior with $M=20$]
 	{\includegraphics[width=6.75cm]{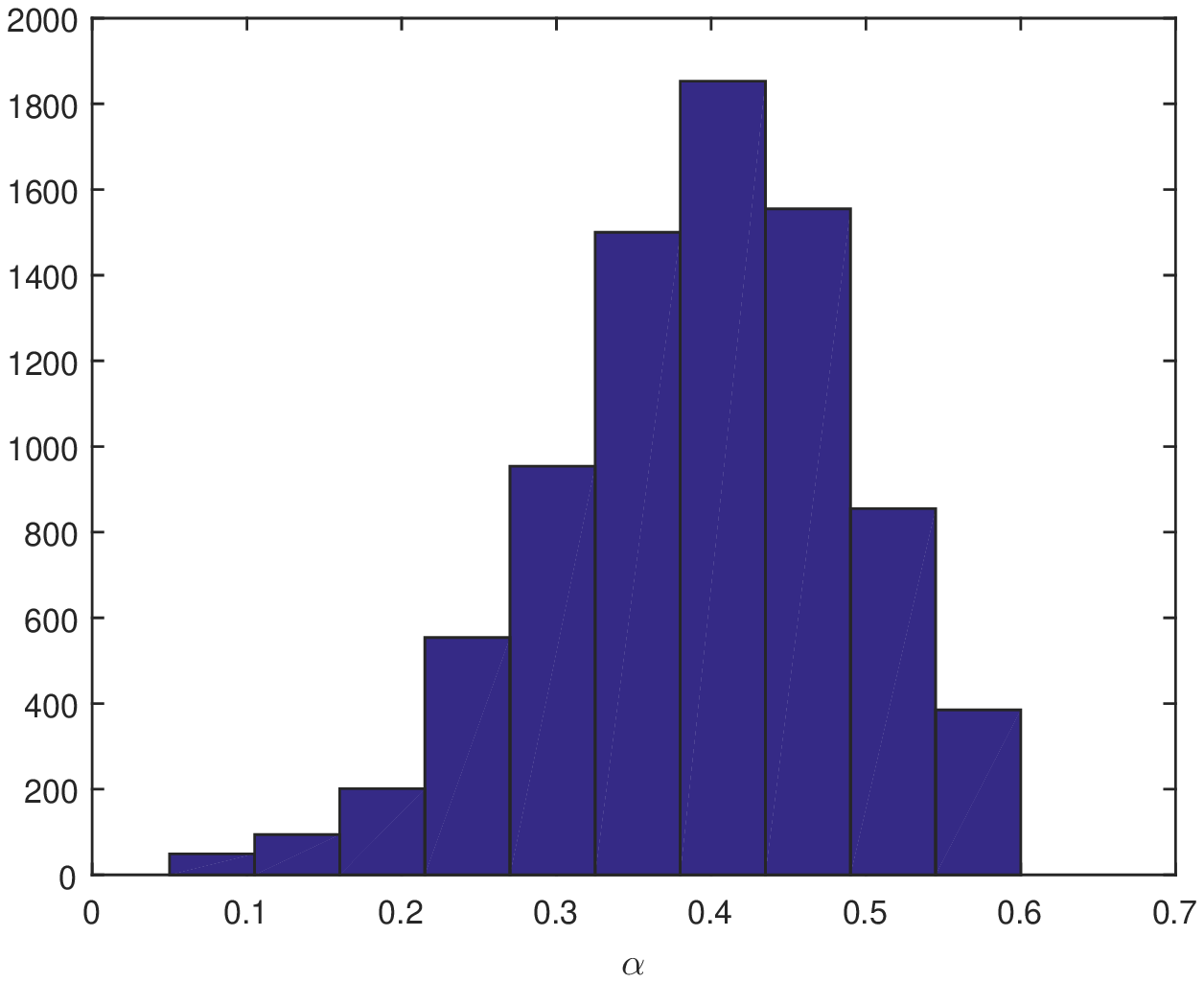}}
   	\caption{Posterior samples (left) and histograms (right) of the analysis of an i.i.d. sample of size $n=100$ from a Yule--Simon distribution with $\alpha=0.40$. From top to bottom, we have Jeffreys prior, loss-based prior with $M=10$ and loss-based prior with $M=20$.}
\label{PosA040}
\end{figure}
Figure \ref{PosA040} shows the posterior samples and posterior histograms derived by applying the Jeffreys prior and the loss-based prior with two different discretizations, that is $M=10$ and $M=20$. The summary statistics of the three posteriors are reported in Table \ref{T40}, where we have the mean, the median, and the $95\%$ credible interval.
\begin{table}[h!]
\centering
\begin{tabular}{cccc}
\hline
Prior & Mean & Median & $95\%$  C.I. \\
\hline
Jeffreys & 0.40& 0.41 &(0.23,0.53) \\
Loss-based $(M=10)$ & 0.40 & 0.4 & (0.2,0.5) \\
Loss-based $(M=20)$ & 0.40 & 0.41 & (0.22,0.56) \\
\hline
\end{tabular}
\caption{Summary statistics of the posterior distributions for the parameter $\alpha$ of the simulated data from a Yule-Simon distribution with $\alpha=0.40$.}
\label{T40}
\end{table}
By comparing the mean of the posterior distributions, we see that they are all centered around the true parameter value.  The credible interval yielded by the loss-based priors with the most dense discretization ($M=20$) is larger than the other two intervals. However, the difference is very small and we can conclude that the three prior distributions result in posteriors which carry the same uncertainty. In other words, the three objective priors perform in the same way.
\begin{figure}[h!]
\centering
	\subfigure[]%Jeffreys]
   {\includegraphics[width=6.75cm]{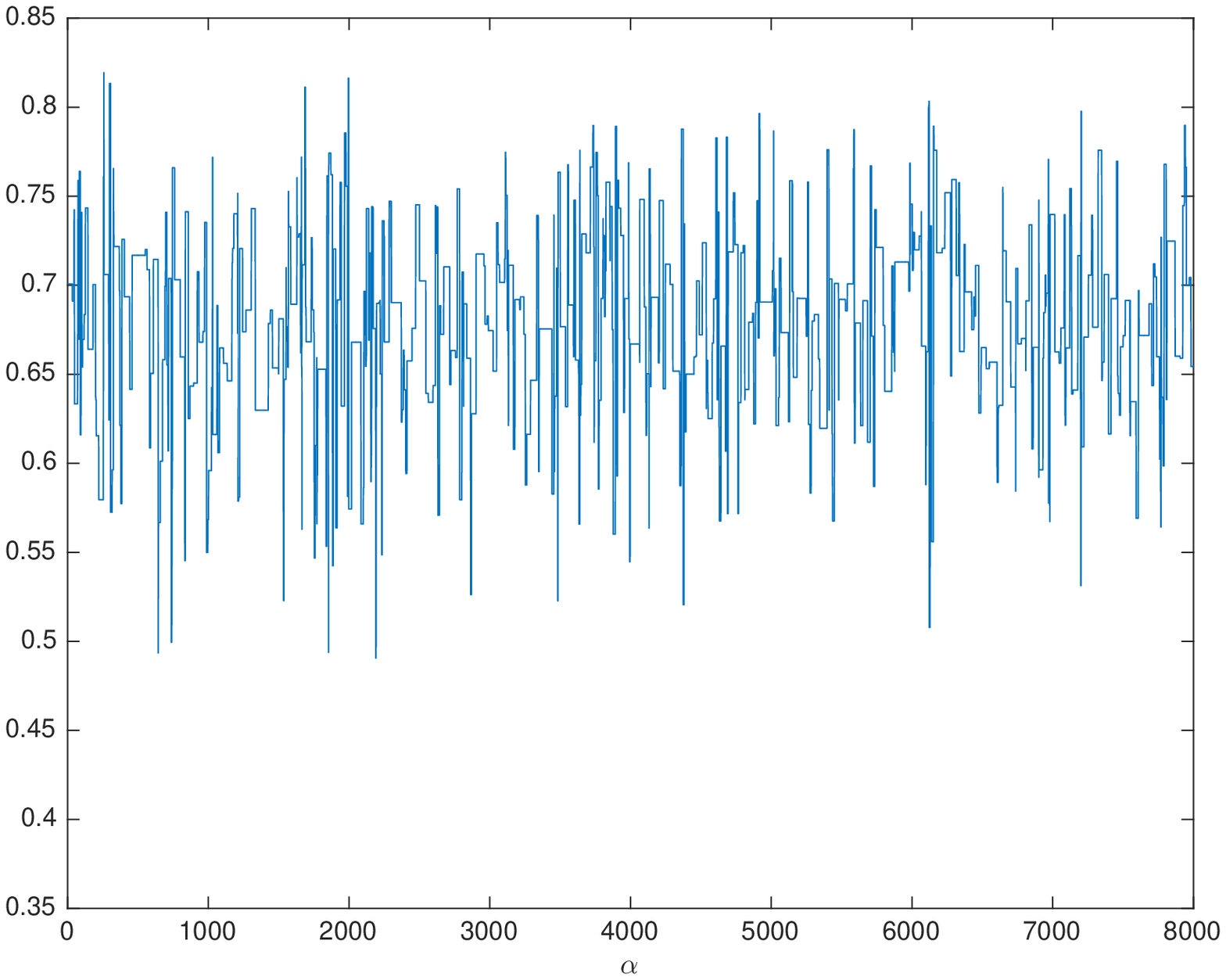}}
	\subfigure[]%Jeffreys]
  	{\includegraphics[width=6.75cm]{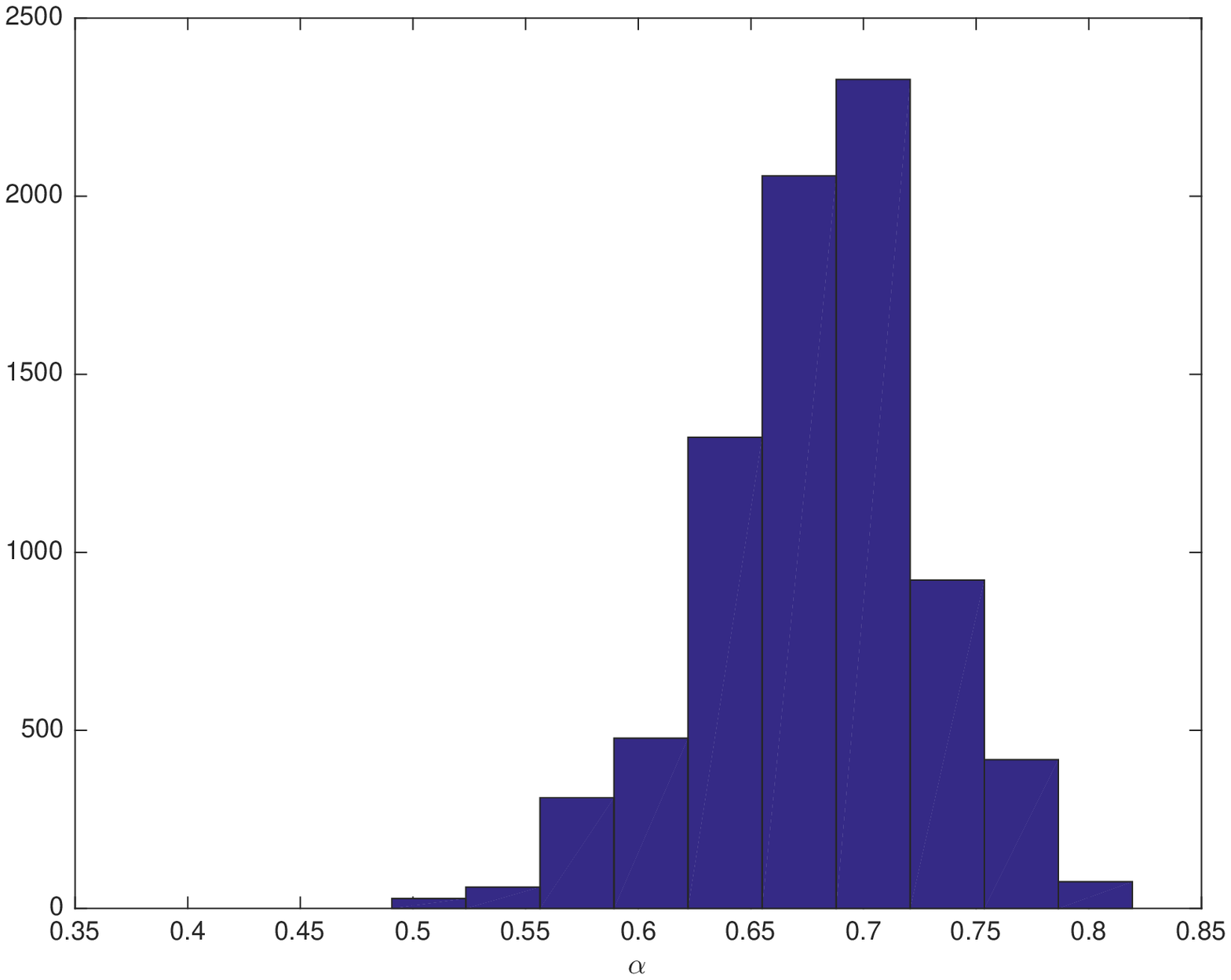}}
  	\subfigure[]%loss-prior with $M=10$]
   {\includegraphics[width=6.75cm]{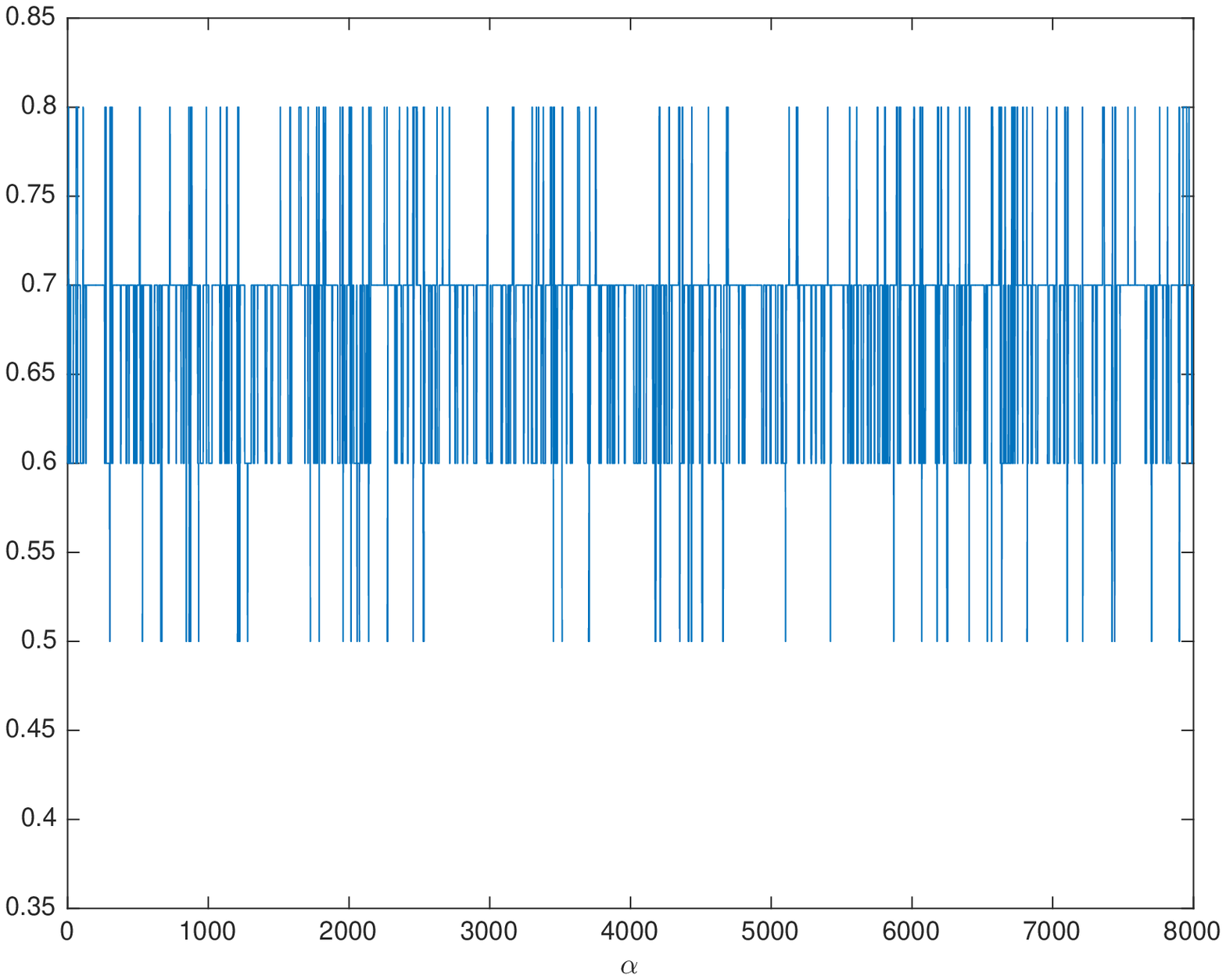}}
   \subfigure[]%loss-prior with $M=10$]
   {\includegraphics[width=6.75cm]{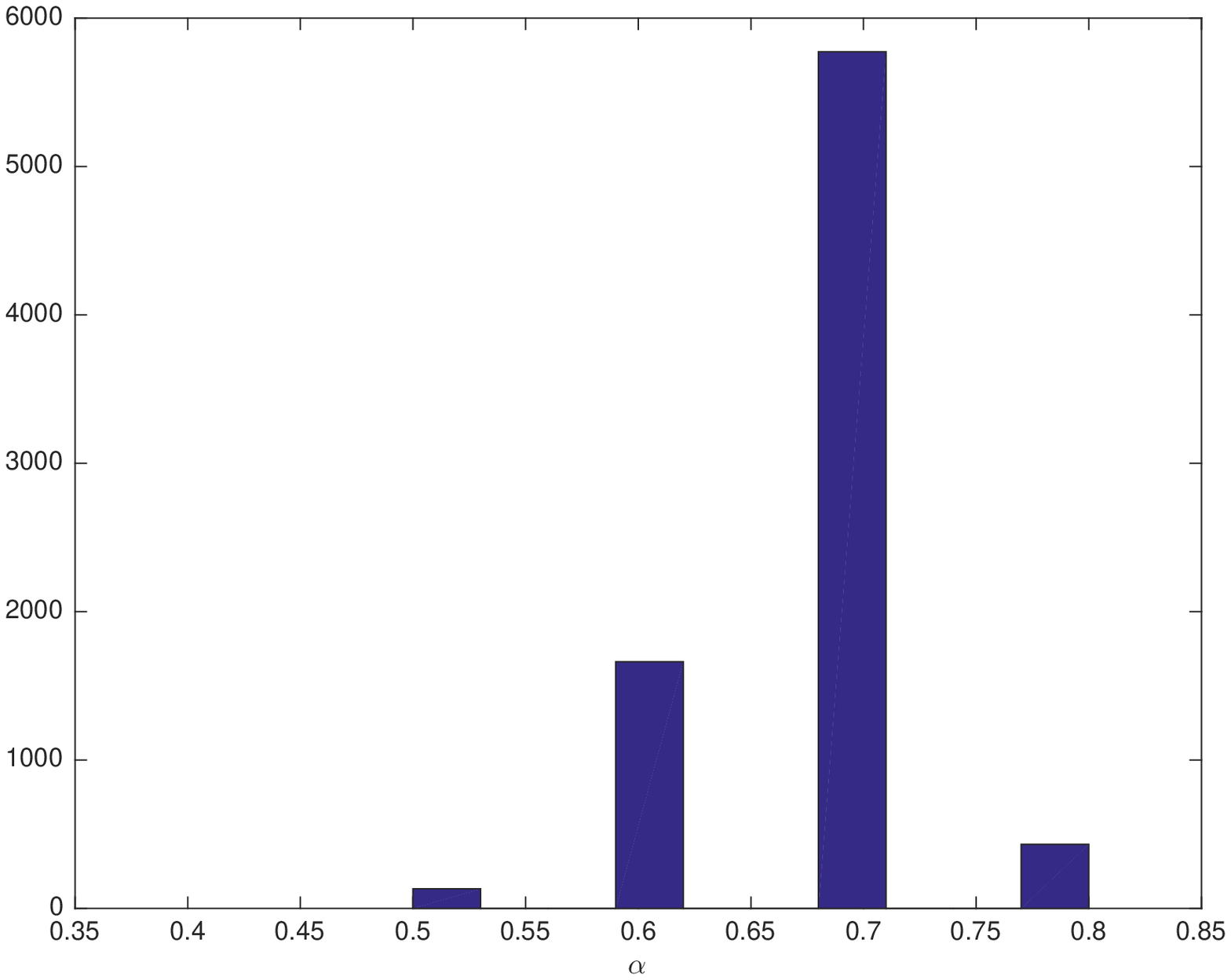}}
   \subfigure[]%loss-prior with $M=20$]
   {\includegraphics[width=6.75cm]{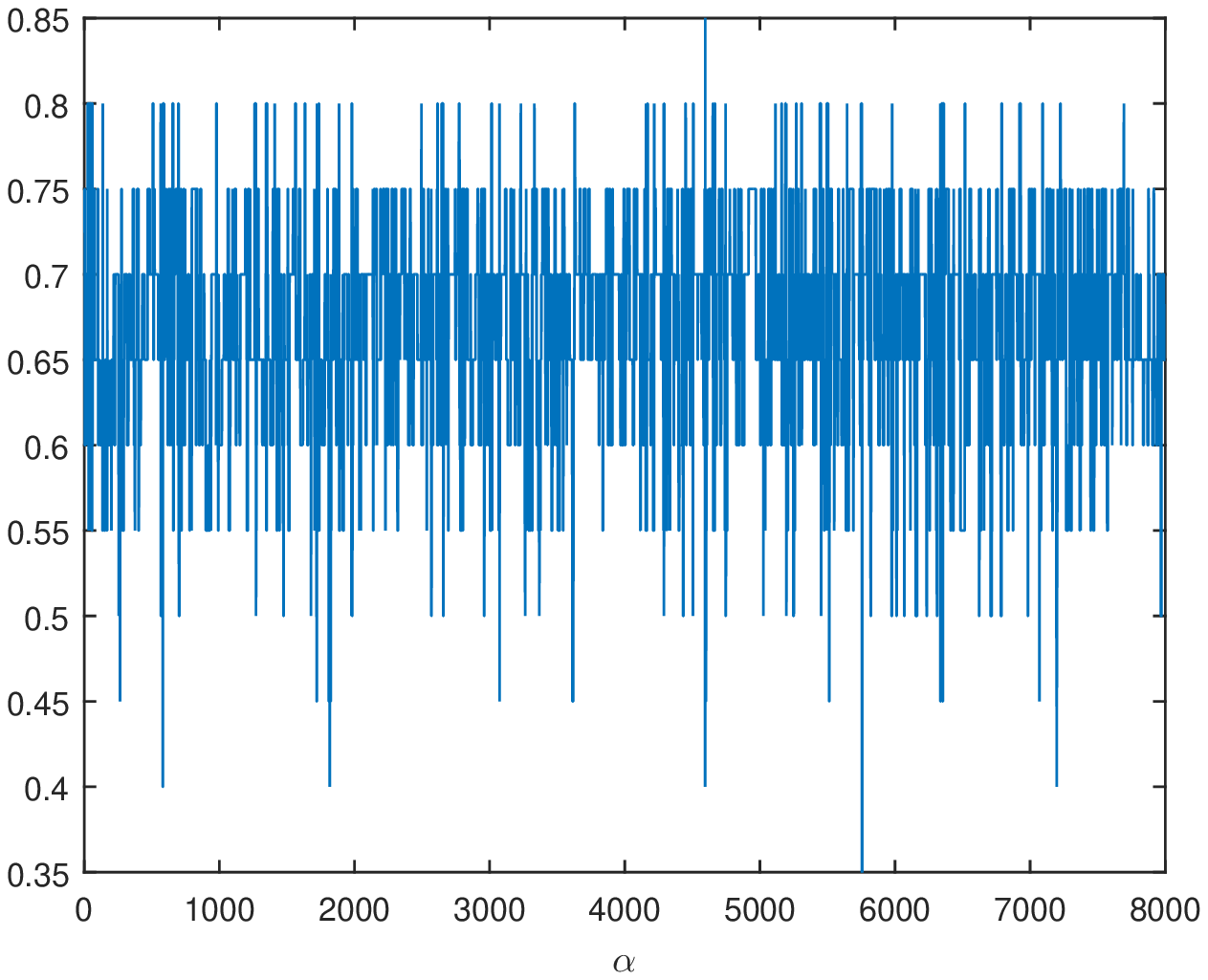}}
    \subfigure[]%loss-prior with $M=20$]
   {\includegraphics[width=6.75cm]{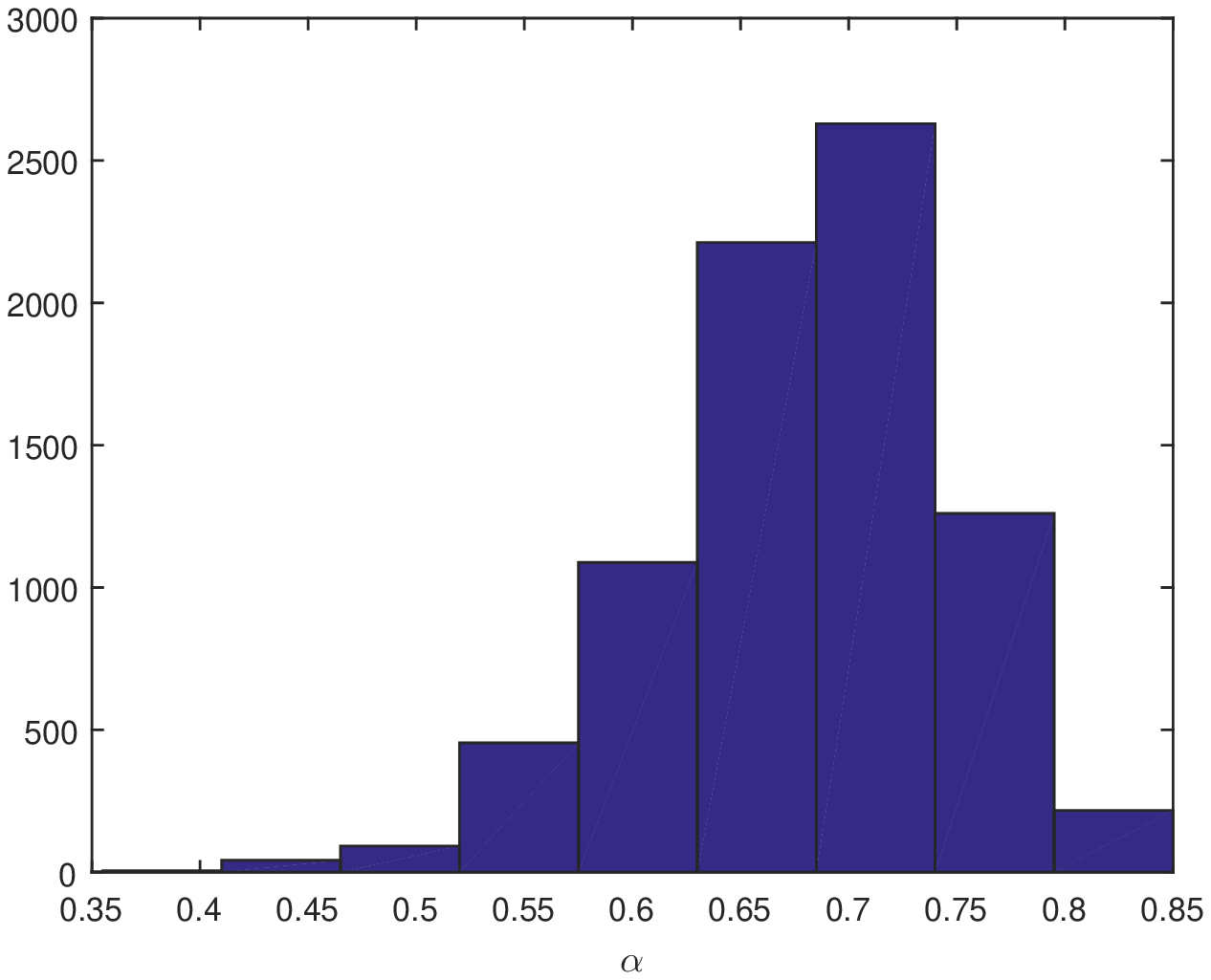}}
\caption{Posterior samples (left) and histograms (right) of the analysis of an i.i.d. sample of size $n=100$ from a Yule--Simon distribution with $\alpha=0.68$. From top to bottom, we have Jeffreys prior, loss-based prior with $M=10$ and loss-based prior with $M=20$.}
\label{PosA068}
\end{figure}

\begin{table}
\centering
\begin{tabular}{cccc}
\hline
Prior & Mean & Median & $95\%$  C. I. \\
\hline
Jeffreys & 0.68 & 0.68 & (0.57,0.77) \\
Loss-based ($M=10)$ & 0.68 & 0.7 & (0.6,0.8) \\
Loss-based ($M=20)$ & 0.68 & 0.68 & (0.55,0.79) \\
\hline
\end{tabular}
\caption{Summary statistics of the posterior distributions for the parameter $\alpha$ of the simulated data from a Yule-Simon distribution with $\alpha=0.68$.}
\label{T68}
\end{table}
Similar considerations can be made for the case where we have sampled $n=100$ observations from a Yule--Simon distribution with $\alpha=0.68$. By inspecting Figure \ref{PosA068} and Table \ref{T68}, we note a very similar behaviour of the three priors, in the sense that the posterior distributions are still centered around the true value of $\alpha$ and that the credible intervals do not present important differences. Note that the choice of a true parameter value which would have not been included in any of the two discretized sample spaces, upon which the loss-prior is based, allows to show that the inferential process appears to be not affected by the discretization, hence motivating it.

To conclude, the simulation study shows no tangible differences in the performance of the prior distributions, in the spirit of objective Bayesian analysis.

%%%%%%%%%%%%%%%%%%%%%%%%%%%%%%%%%%%%%%%%%%%%%%%%%%%%%%
%% 	 	Real Data 		%%%
%%%%%%%%%%%%%%%%%%%%%%%%%%%%%%%%%%%%%%%%%%%%%%%%%%%%%%
\section{Real Data Application}
\label{RealData}
To illustrate the proposed priors, both the Jeffreys and the loss-based prior for the Yule-Simon distribution, we analyze three datasets. The first dataset concerns daily increments of four popular social networks stock indexes in the  US market, the second contains the frequencies of surnames observed in the $1990$ US Census, and the last dataset consists of 'number one' hits in the US  music industry.

%\url{http://www.census.gov/en.html}}

\subsection{Social network stock indexes}
\label{Stock}
We analyze different data in the social media marketing, in particular we focus on Facebook, Twitter, Linkedin and Google. These four major companies are the most powerful social networks in the world and are listed in the Wall Street exchange market (\href{http://finance.yahoo.com}{http://finance.yahoo.com}). We analyze the daily increments for the stocks and, in particular, we consider the adjusted closing price from the $1^{st}$ of October 2014 to the $11^{th}$ of March 2016, for a total of $n=365$ observations. The daily increments are obtained by applying $z_{t}=\left|r_{t}/r_{t-1}-1\right|\cdot 100$, for $t=2,\dots, 365$, where $r_{t}$ is the adjusted closing price for the index at day $t$, and we built our frequency on it. These are shown in Figure \ref{DayliIncre}, while Figure \ref{DicretHists} shows the histogram of the frequencies of the discretized data. The discretization has been done by counting the number of times a daily return took a value truncated at the second decimal digit. For example, if two observed daily returns are $1.2494$ and $1.2573$, they were both considered as two occurrences of the same value.
\begin{figure}[h!]
\centering
\subfigure[Facebook]
  	{\includegraphics[width=6.75cm]{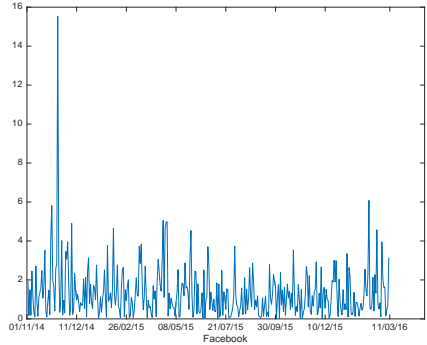}}
   \subfigure[Google]
   {\includegraphics[width=6.75cm]{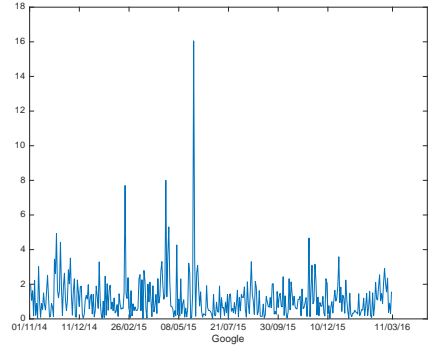}}
   \subfigure[Linkedin]
   {\includegraphics[width=6.75cm]{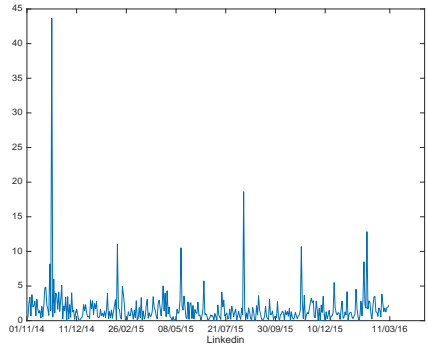}}
   \subfigure[Twitter]
   {\includegraphics[width=6.75cm]{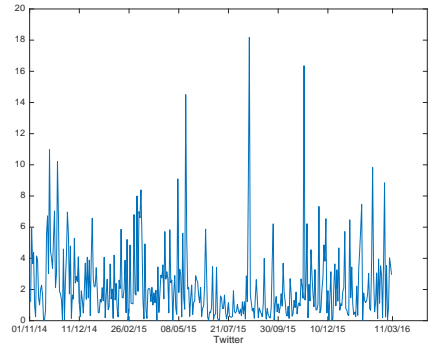}}
   \caption{Daily increments for Facebook, Google, Linkedin and Twitter from the $1^{st}$ of October 2014 to the $11^{th}$ of March 2016.}
   \label{DayliIncre}
   \end{figure}
By inspecting the histograms in Figure \ref{DicretHists} is seems that the (transformed) Yule--Simon distribution might be a suitable statistical model to represent the data. We apply the Bayesian framework and obtain the posterior distribution for the parameter of interest as
$$\pi(\alpha|\textbf{k}) \propto L(\textbf{k}|\alpha)\pi(\alpha),$$
where $\textbf{k}=(k_1,\dots,k_n)$ represents the set of observations, i.e. the frequencies of the discretized daily returns, $L(\textbf{k}|\alpha)$ the likelihood function and $\pi(\alpha)$ the prior distribution which, in turn, has the form of the Jeffreys prior in \eqref{JefPrior} or the loss-based prior \eqref{OBB}.
\begin{figure}[h!]
\centering
\subfigure[Facebook]
  	{\includegraphics[width=6.75cm]{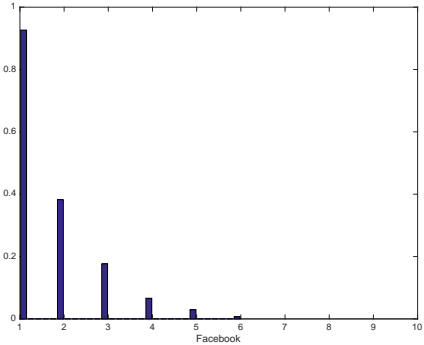}}
   \subfigure[Google]
   {\includegraphics[width=6.75cm]{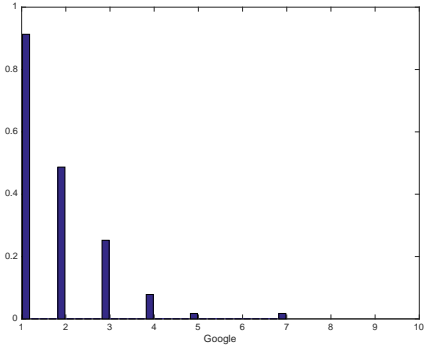}}
   \subfigure[Linkedin]
   {\includegraphics[width=6.75cm]{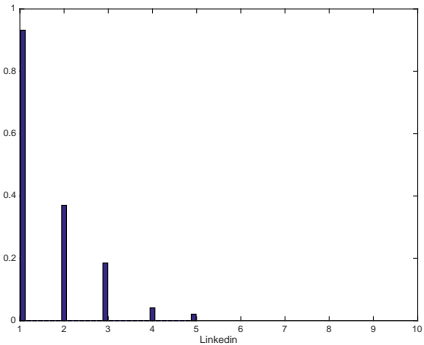}}
   \subfigure[Twitter]
   {\includegraphics[width=6.75cm]{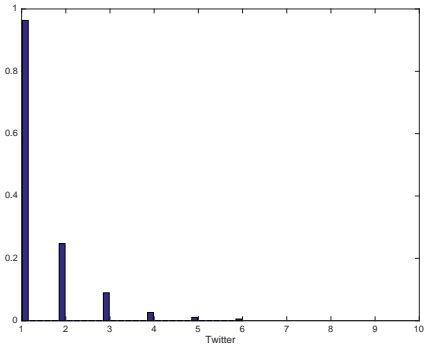}}
   \caption{Histograms of the discretized daily returns for Facebook, Google, Linkedin and Twitter.}
   \label{DicretHists}
   \end{figure}
We have obtained the posterior distributions for the parameter $\alpha$ of the transformed Yule-Simon distribution by Monte Carlo methods. We run 25,000 iterations with a burn-in period of $5,000$ iterations. We have reported the chain and the histogram of the posterior distributions in Figure \ref{Facebook} and in Figure \ref{Google}, with the corresponding summary statistics in Table \ref{T2}. Note that, with the purpose of limiting the amount of space used, we have included the plots of the Facebook and Google daily returns only.

\begin{figure}[h!]
\centering
	\subfigure[]
  	{\includegraphics[width=6.75cm]{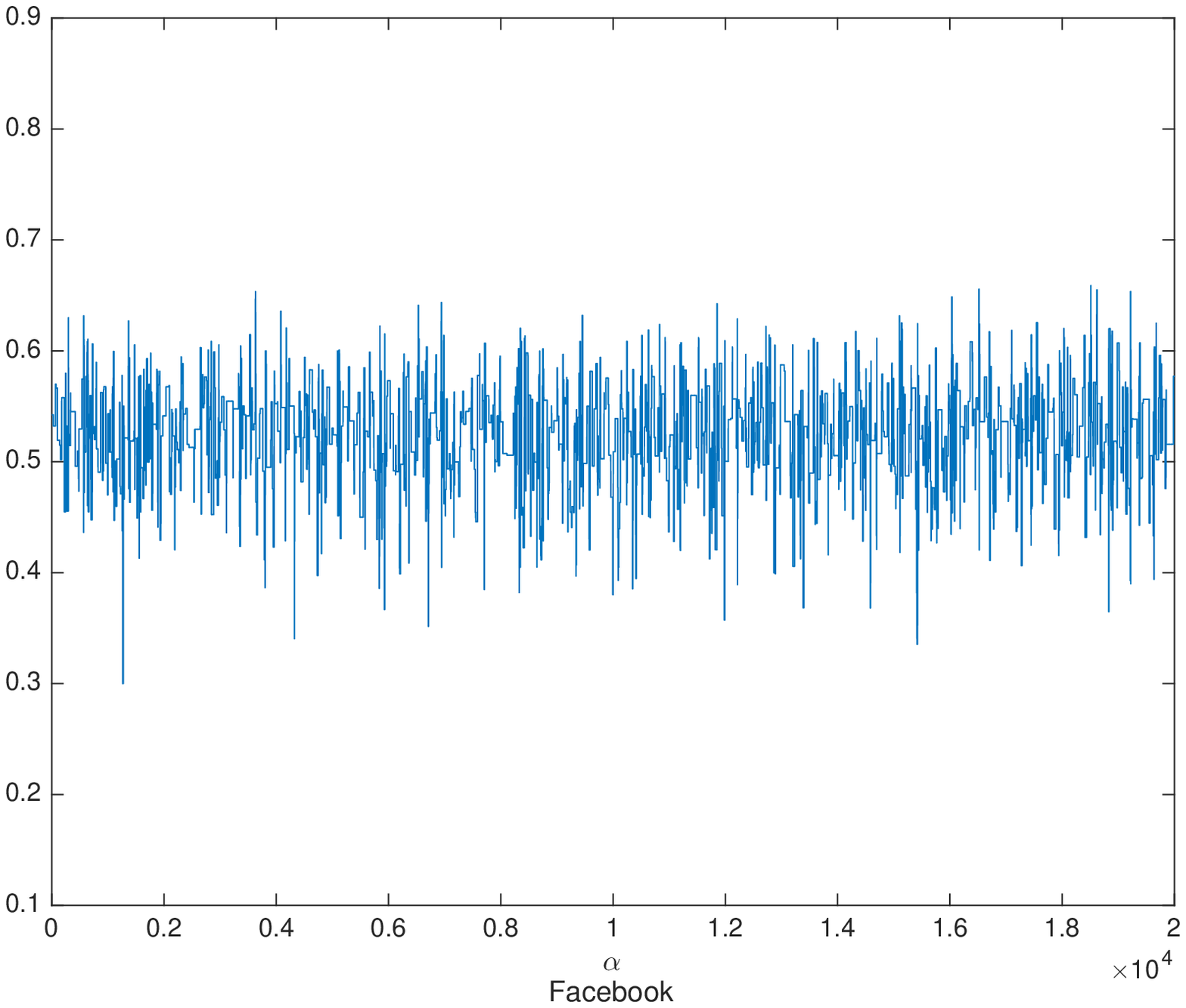}}
  	\subfigure[]
   	{\includegraphics[width=6.75cm]{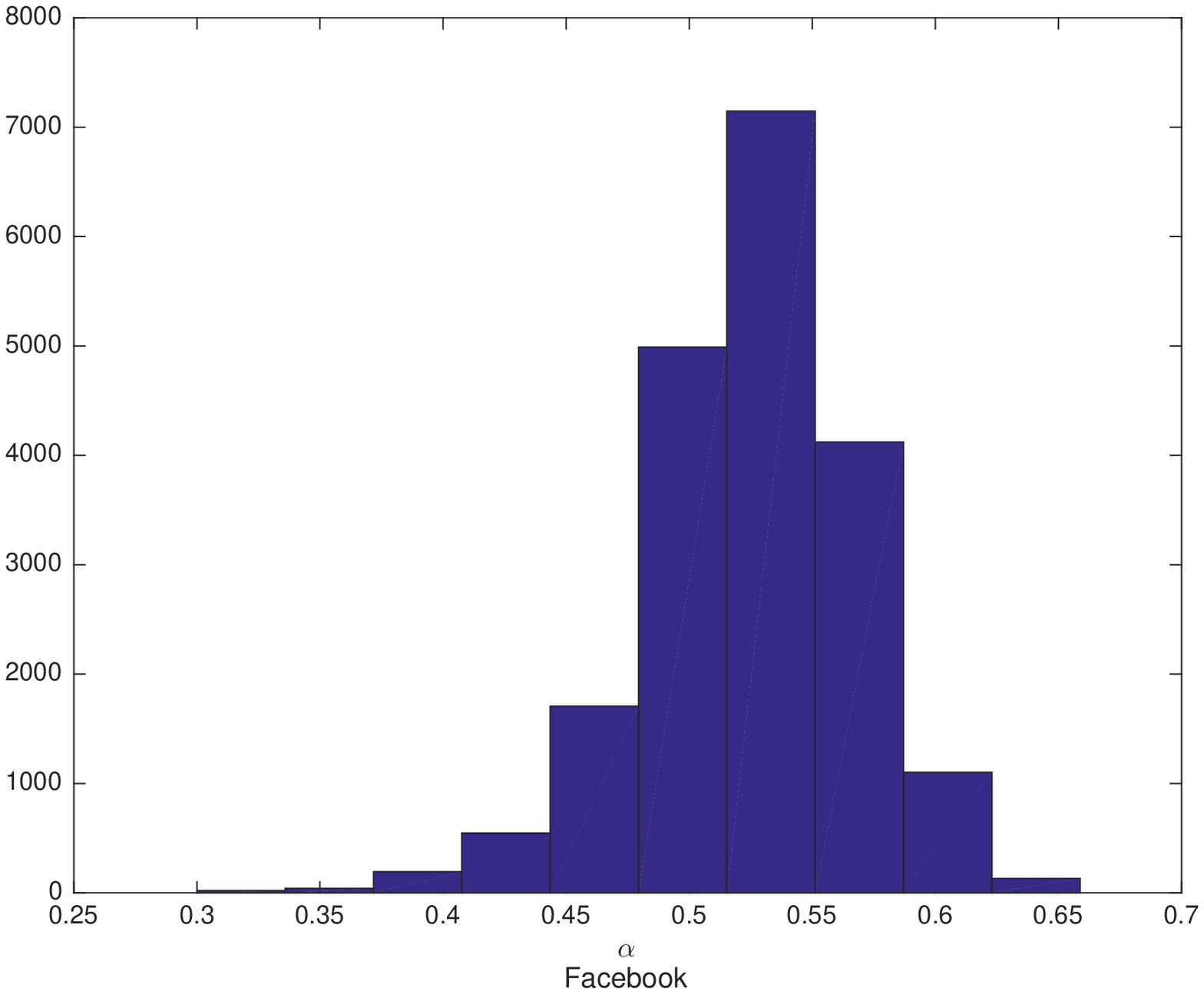}}
  	\subfigure[]
   	{\includegraphics[width=6.75cm]{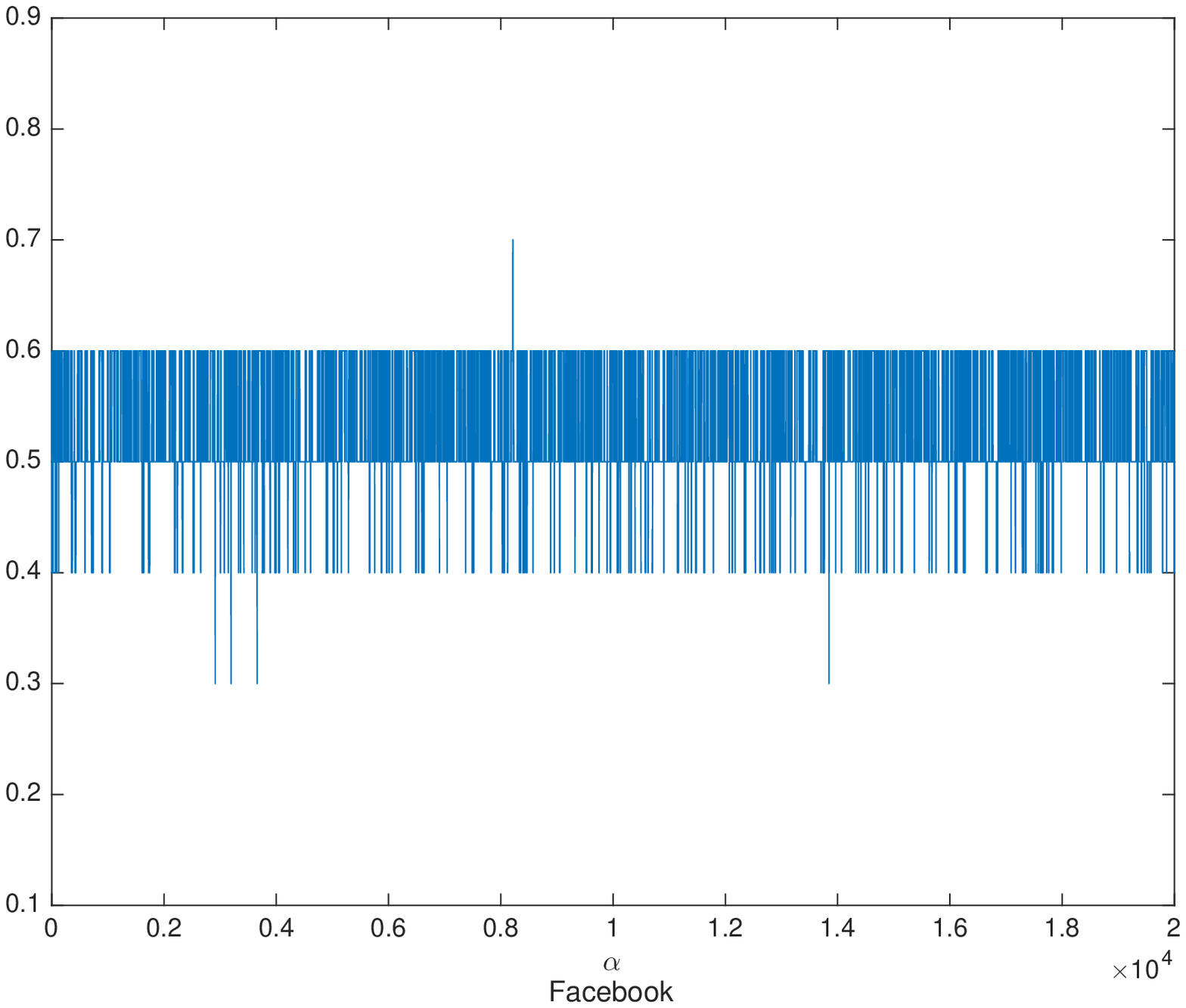}}
   	\subfigure[]
   	{\includegraphics[width=6.75cm]{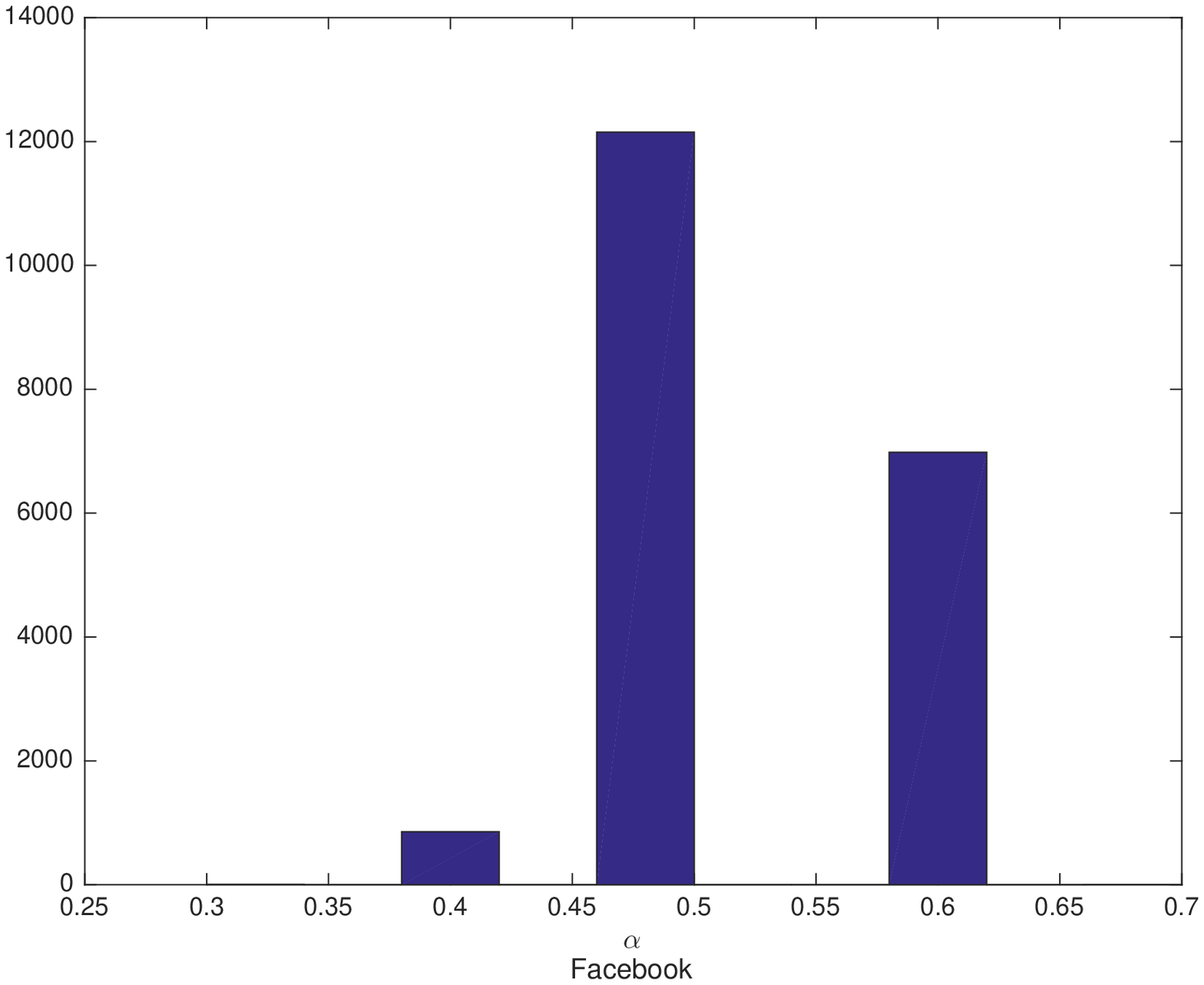}}
   	\subfigure[]
   	{\includegraphics[width=6.75cm]{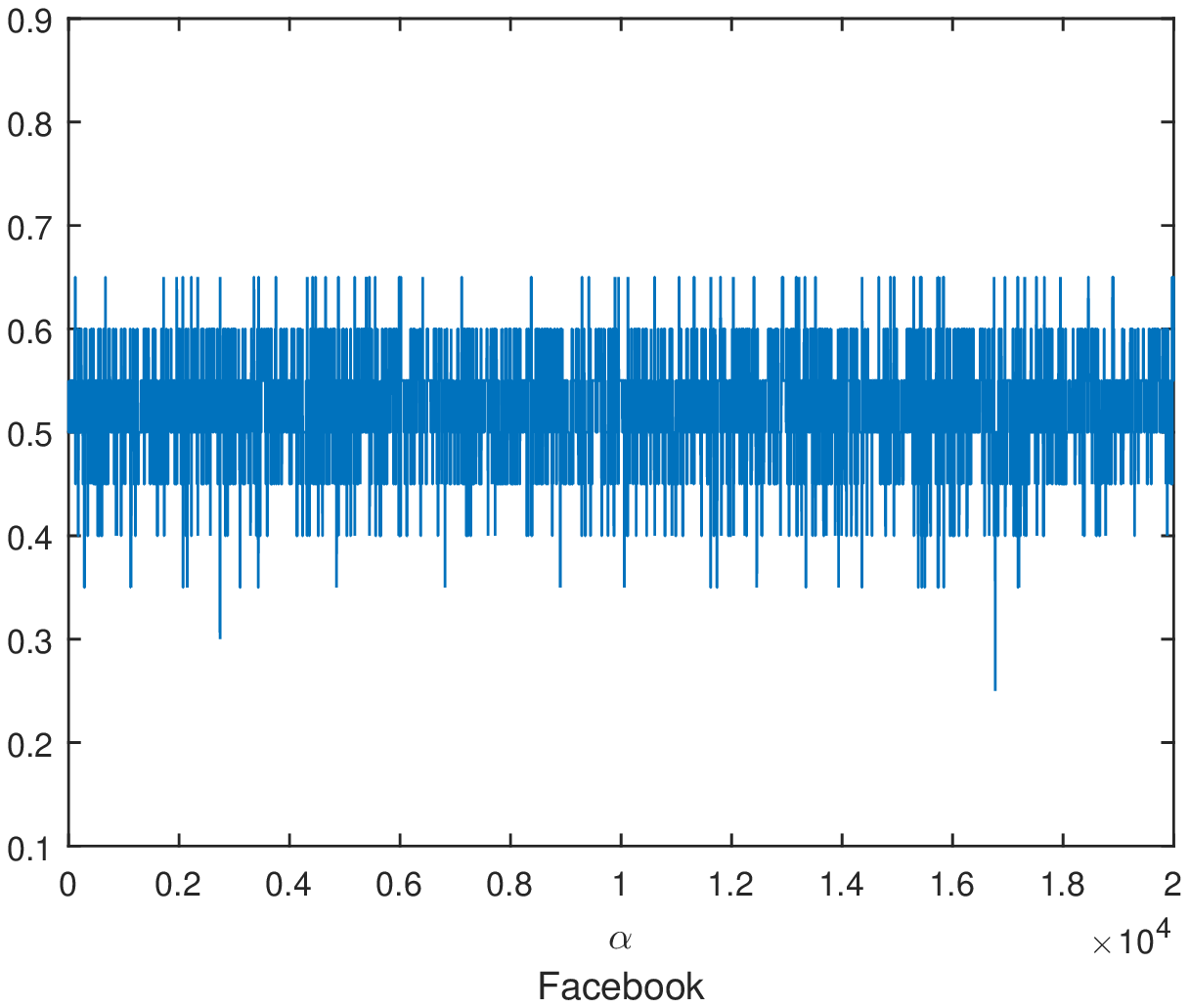}}
    \subfigure[]
   	{\includegraphics[width=6.75cm]{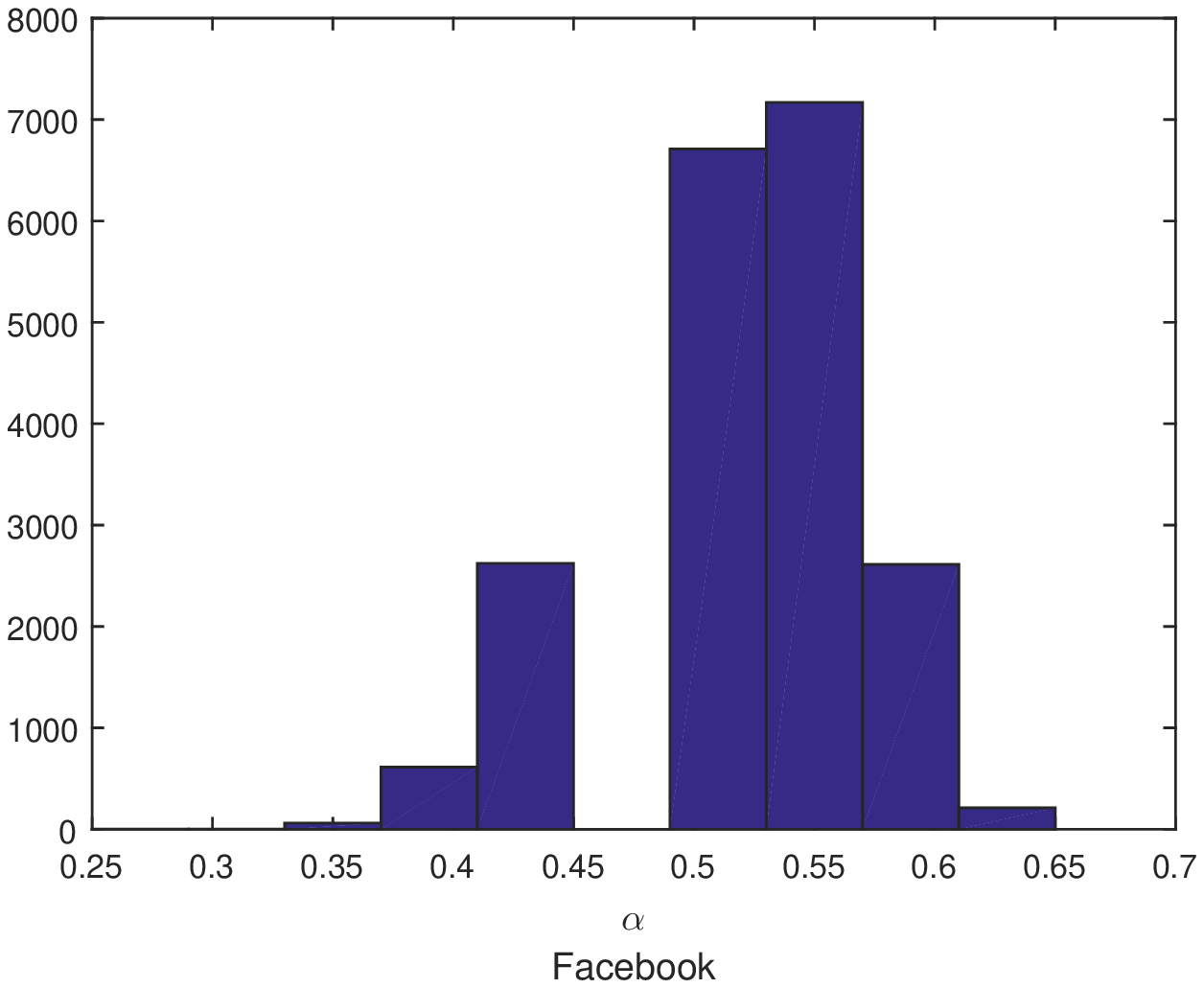}}
\caption{Posterior samples (left) and posterior histograms (right) for the Facebook daily returns obtained by applying the Jeffreys prior (top), the loss-based prior with $M=10$ (middle) and the loss-based prior with $M=20$ (bottom).}
\label{Facebook}
\end{figure}

\begin{figure}[h!]
\centering
	\subfigure[]
  	{\includegraphics[width=6.75cm]{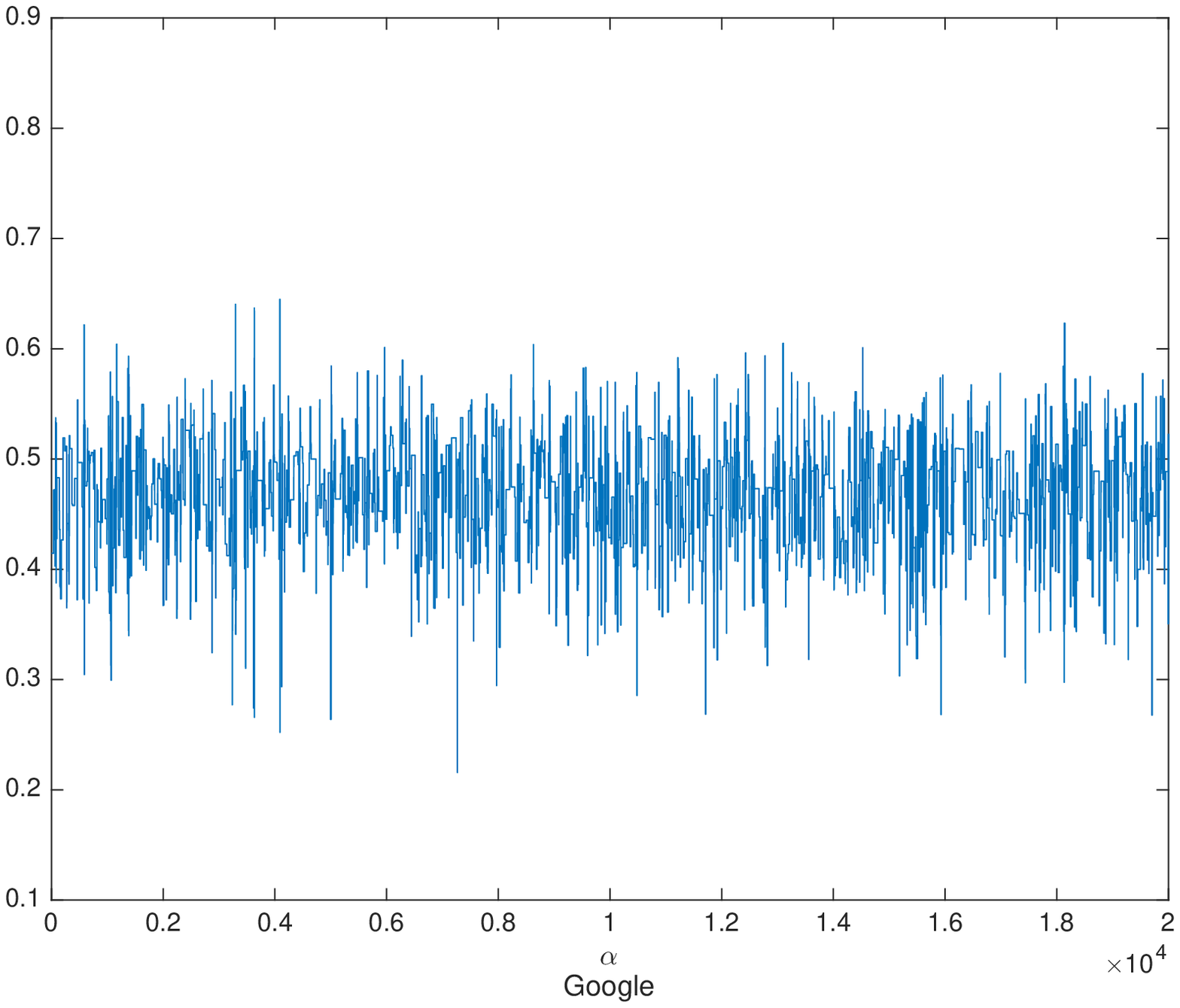}}
  	\subfigure[]
   {\includegraphics[width=6.75cm]{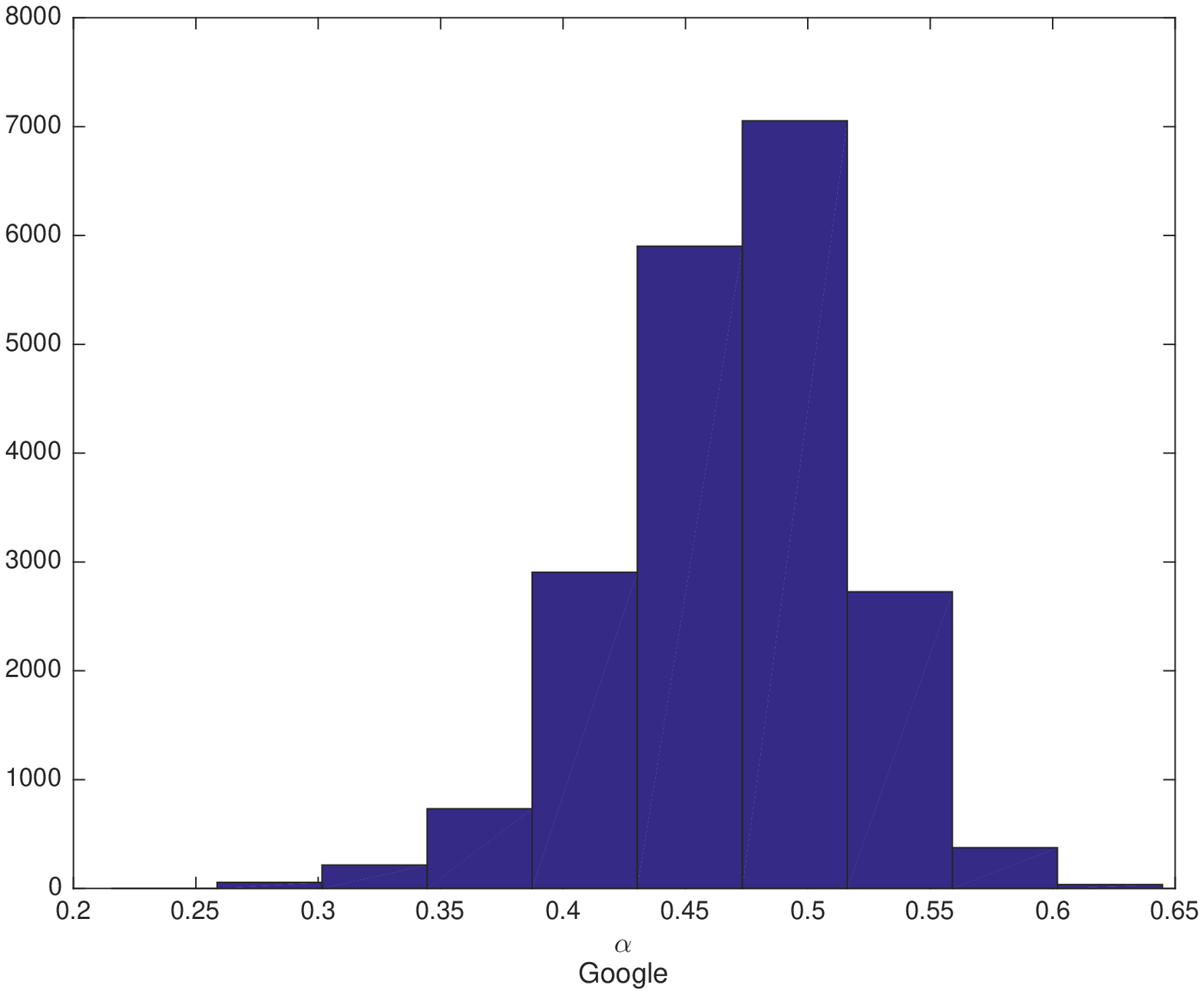}}
   \subfigure[]
   {\includegraphics[width=6.75cm]{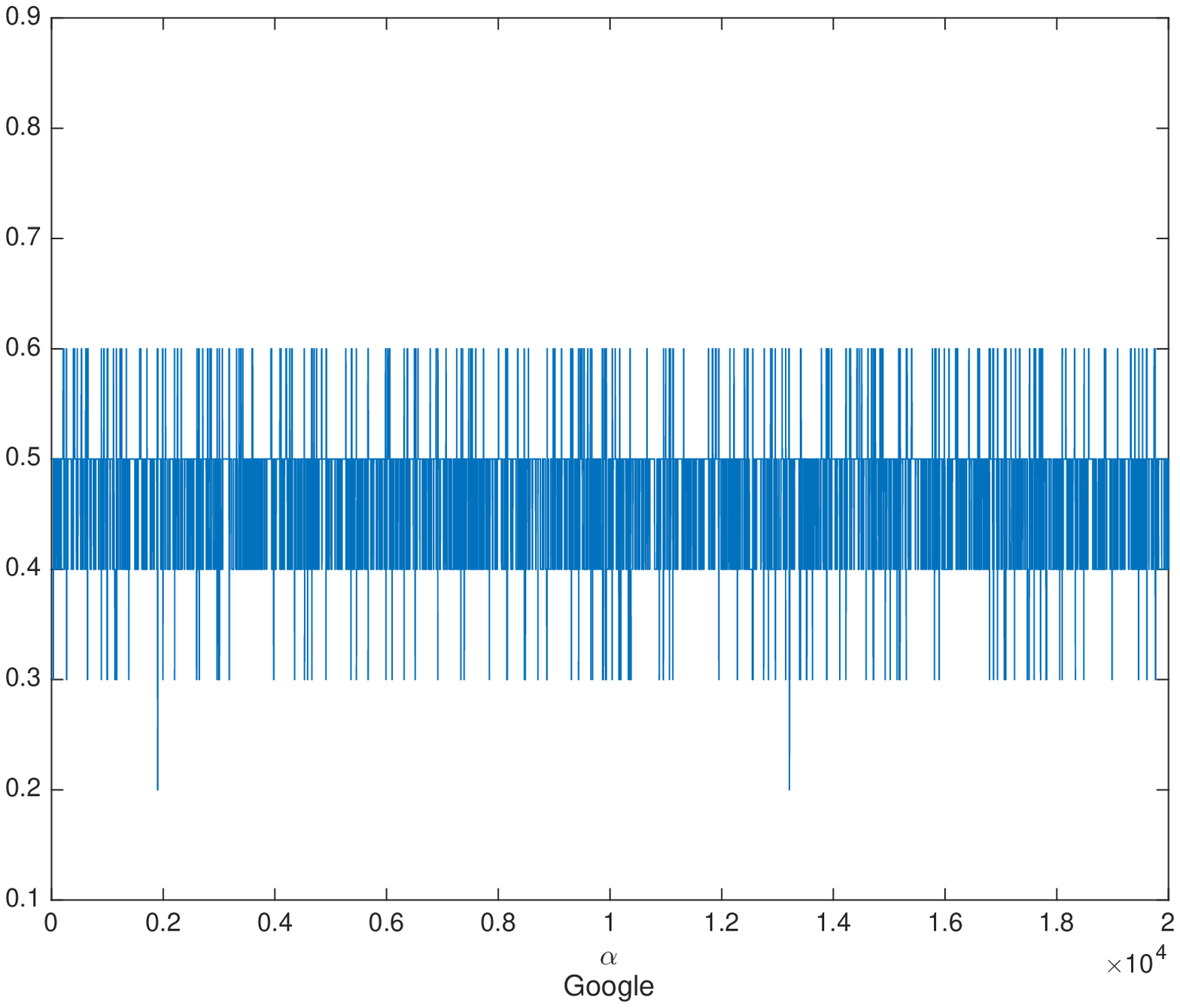}}
   \subfigure[]
   {\includegraphics[width=6.75cm]{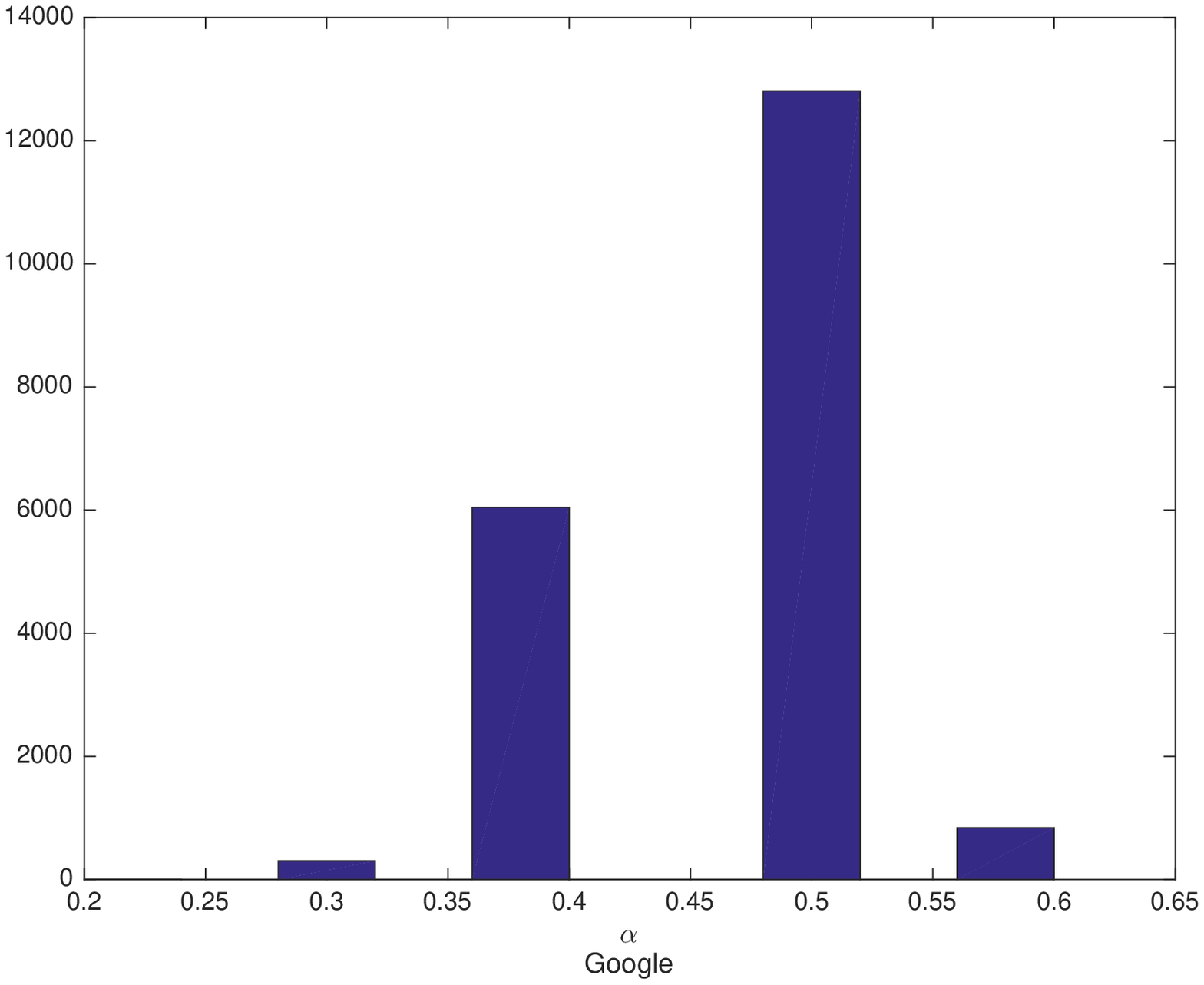}}
   \subfigure[]
   {\includegraphics[width=6.75cm]{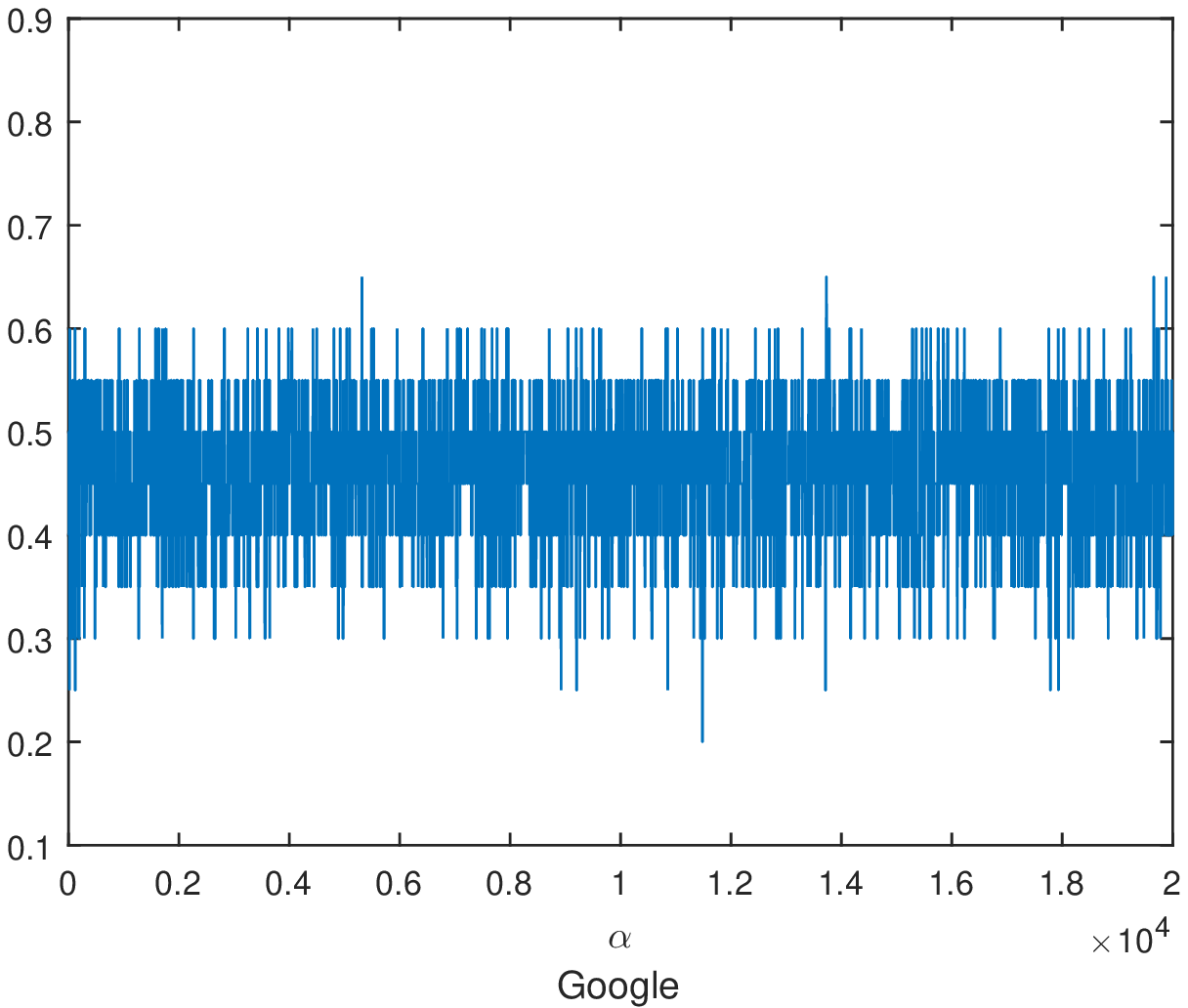}}
   \subfigure[]
   {\includegraphics[width=6.75cm]{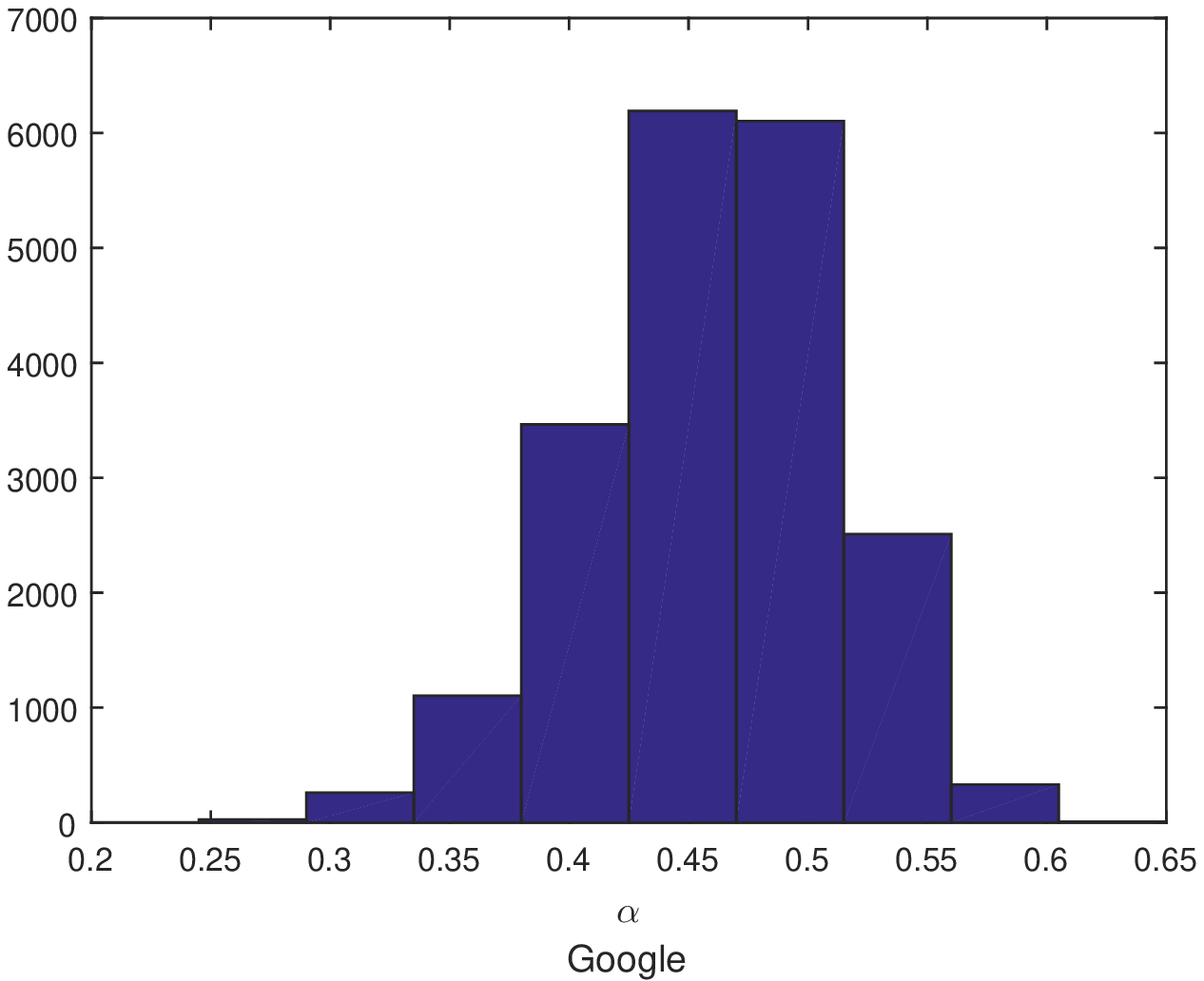}}
\caption{Posterior samples (left) and posterior histograms (right) for the Google daily returns obtained by applying the Jeffreys prior (top), the loss-based prior with $M=10$ (middle) and the loss-based prior with $M=20$ (bottom).}
\label{Google}
\end{figure}

\begin{table}[h!]
\centering
\begin{tabular}{ccccc}
\hline
Company & Prior & Mean & Median & $95\%$  C.I. \\
\hline
Facebook & Jeffreys & 0.53 & 0.53 & (0.43, 0.61) \\
Facebook & Loss-based $(M=10)$ & 0.53 & 0.5 & (0.4, 0.6) \\
Facebook & Loss-based $(M=20)$ & 0.52 & 0.55 & (0.40, 0.60) \\
\hline
Google & Jeffreys & 0.47 & 0.47 & (0.37, 0.55) \\
Google & Loss-based $(M=10)$ & 0.47 & 0.5 & (0.4, 0.6) \\
Google & Loss-based $(M=20)$ & 0.47 & 0.46 & (0.35, 0.55) \\
\hline
Linkedin & Jeffreys & 0.56 & 0.57 & (0.47, 0.64) \\
Linkedin & Loss-based $(M=10)$ & 0.57 & 0.6 & (0.5, 0.6) \\
Linkedin & Loss-based $(M=20)$ & 0.56 & 0.55 & (0.45, 0.65) \\
\hline
Twitter & Jeffreys & 0.68 & 0.68 & (0.62, 0.73) \\
Twitter & Loss-based $(M=10)$ & 0.69 & 0.7 & (0.6, 0.7) \\
Twitter & Loss-based $(M=20)$ & 0.68 & 0.70 & (0.60, 0.75) \\
\hline
\end{tabular}
\caption{Summary statistics of the posterior distribution for the parameter $\alpha$ of the social network stock index data.}
\label{T2}
\end{table}
For all the four assets we notice that the results for $\alpha$ are very similar, as can be inferred by the minimal (or absence of) difference between the means and the medians. The credible intervals, as well, are very similar, with a slight larger size for the case where the loss-based prior with ($M=20$) is applied. One way of interpreting the results is as follows. The parameter $\alpha$ can be seen as the probability that the next observation is different from the ones observed so far, and therefore we note that Twitter has the highest chance to take a daily increment not yet observed, while Google has the smallest. 
%It is natural to think that, the larger the sample size, the smaller the value of $\alpha$ which would be estimated. 

\subsection{Census Data - Surname analysis}
\label{Census}

\begin{table}
\begin{center}
\begin{tabular}{cccccc}
\hline
\#& Surname & Frequency & \#& Surname & Frequency \\
\hline
1& Smith & 2502021 & 6 & Davis & 1193807 \\
2& Johnson & 2014550 & 7 & Miller & 1054530 \\
3& Williams & 1738482 & 8 & Wilson & 843126 \\
4& Jones & 1544488 & 9 & Moore & 775975 \\
5& Brown & 1544488 & 10 & Taylor & 773488 \\
\hline
\end{tabular}
\caption{Ten most common Surname in United States from the Census 1990 analysis.}
\label{TCensus}
\end{center}
\end{table}
The second example we examine  concerns with the frequency of surnames in the US (\href{http://www.census.gov/en.html}{http://www.census.gov/en.html}). From the population censuses \citep{Mar10}, we focus on the US Census completed in 1990 and consider the first 500 most common surnames. Refer to Table \ref{TCensus} for a list of the first 10 most frequent surnames.
Briefly, the process followed by \cite{Mar10} to obtain the data converts the surname with Senior (SR), Junior (JR) or a number in the last name field (f.e. Moore Sr or Moore Jr or Moore III are converted to Moore) and, in addition, the authors  examined each name entry for the possibility of an inversion (e.g. a first name appearing in the last name fields or vice-versa). However, as there is the possibility of having many surnames that also inverted can sound absolutely right, the authors considered also the surname of the spouse, obtaining additional information to invert the name field of the entire family.

The analysis has been performed by running both the Markov Chain Monte Carlo for 25,000 iterations, with a burn-in of 5,000 iterations.

The posterior samples and the posterior histograms are shown in Figure \ref{Surname}, with the corresponding summary statistics of the posterior distributions reported in Table \ref{TSurn}. We again notice similarities to the simulation study and the analysis of daily increments, in the sense that means and medians are very similar for each prior, and the $95\%$ credible interval obtained by applying the loss-based prior with $M=20$ is slightly larger than the one obtained by using either the Jeffreys prior or the loss-based prior with $M=10$.
\begin{figure}[h!]
\centering
	\subfigure[]
	{\includegraphics[width=6.75cm]{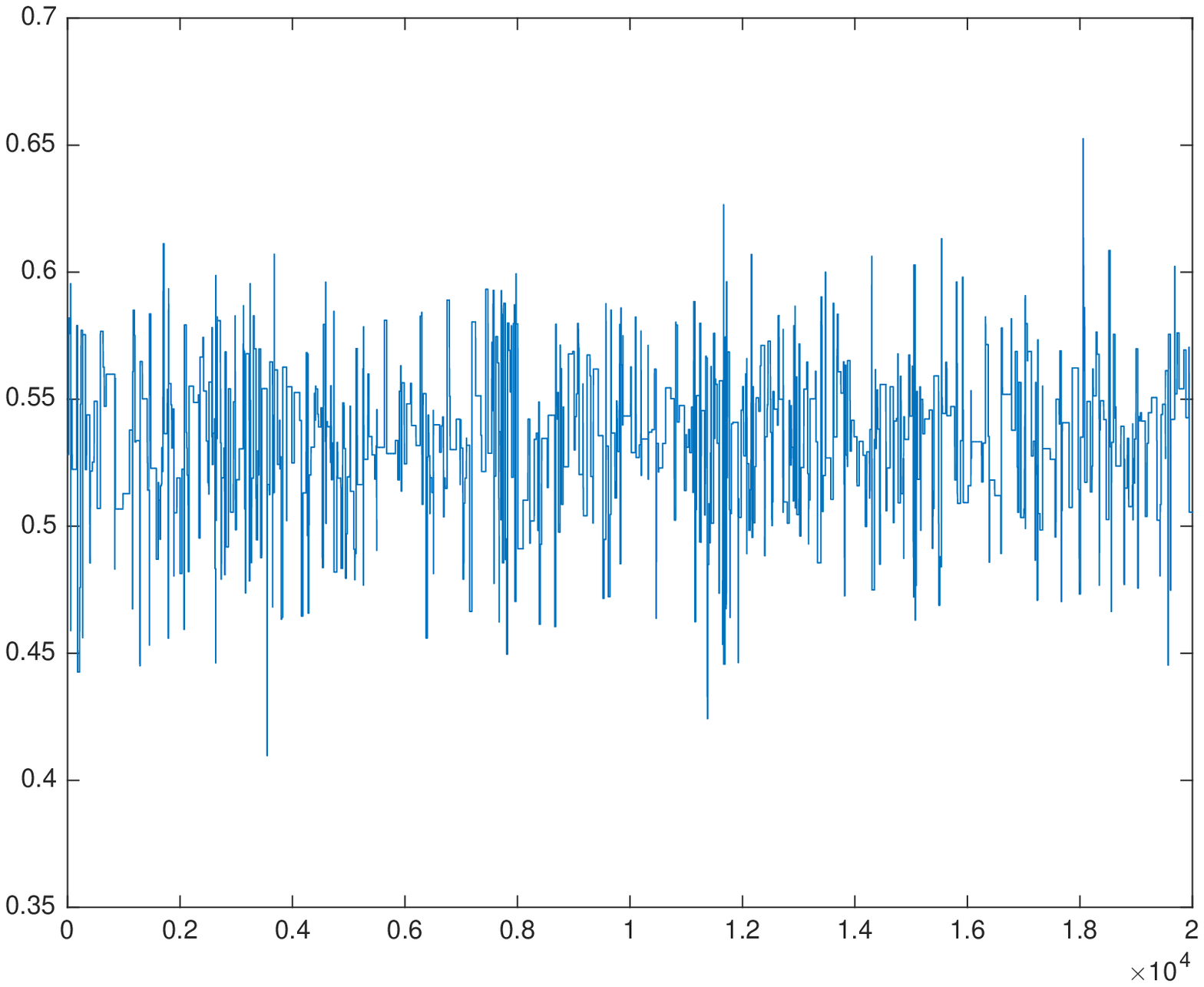}}
	\subfigure[]
  	{\includegraphics[width=6.75cm]{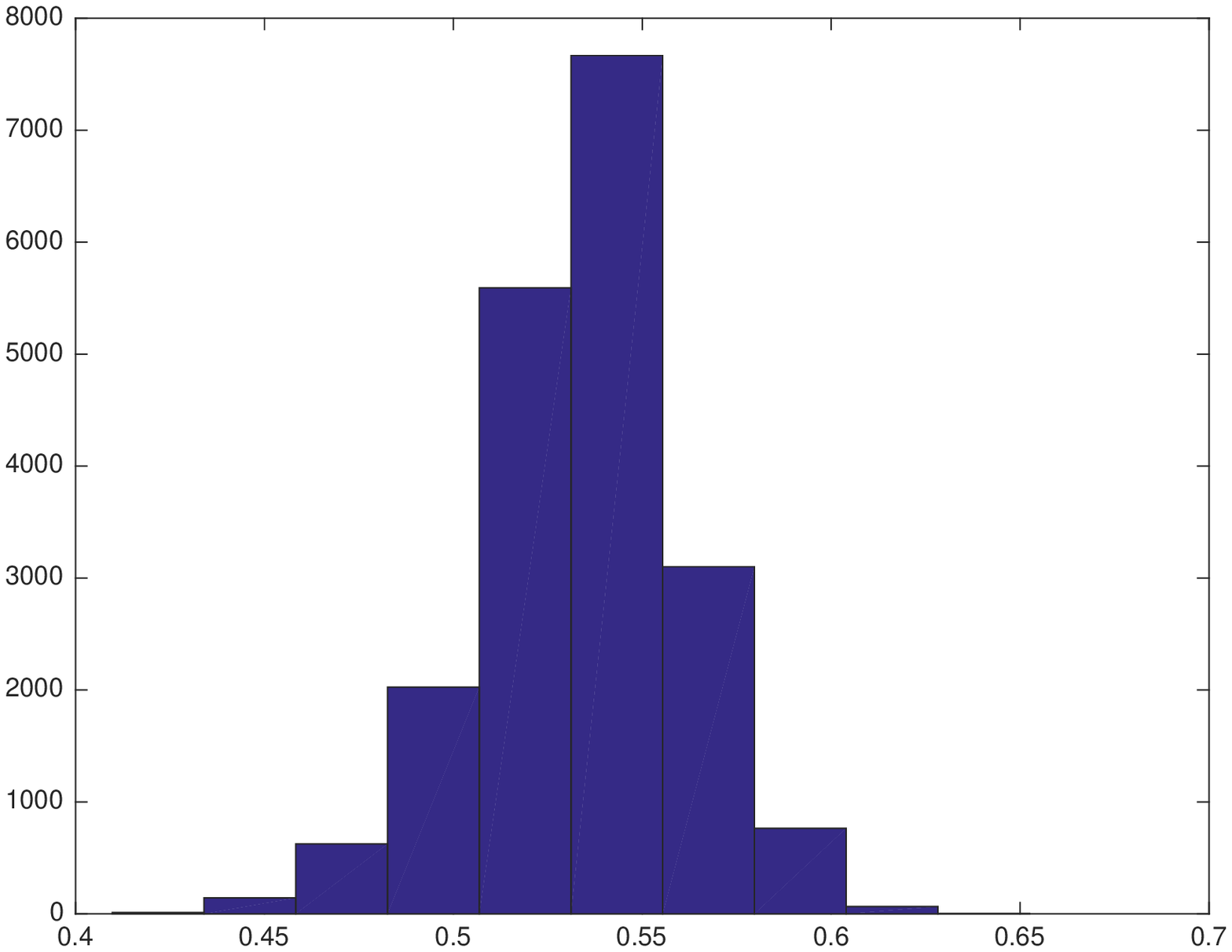}}
  	\subfigure[]
   	{\includegraphics[width=6.75cm]{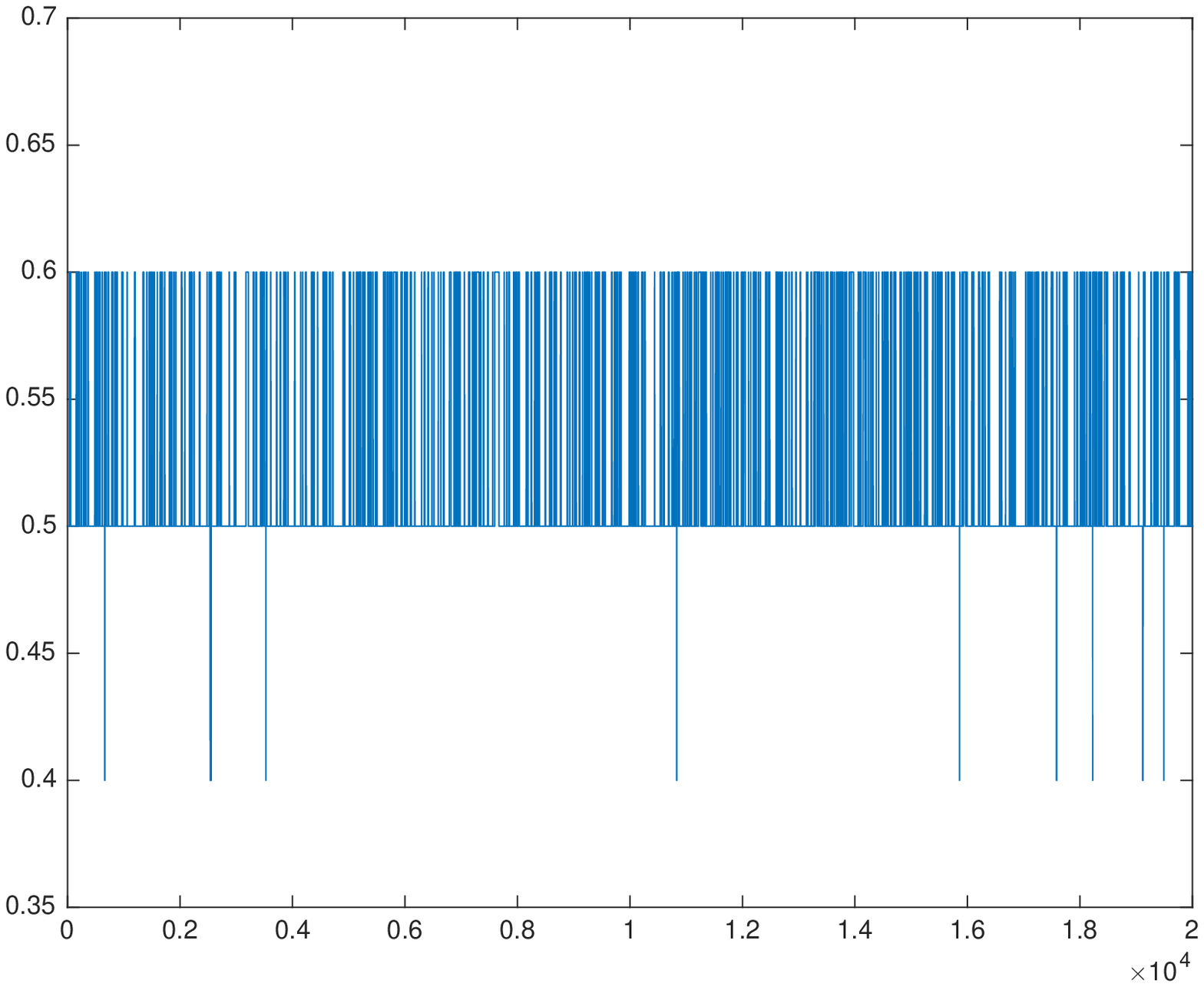}}
   	\subfigure[]
   	{\includegraphics[width=6.75cm]{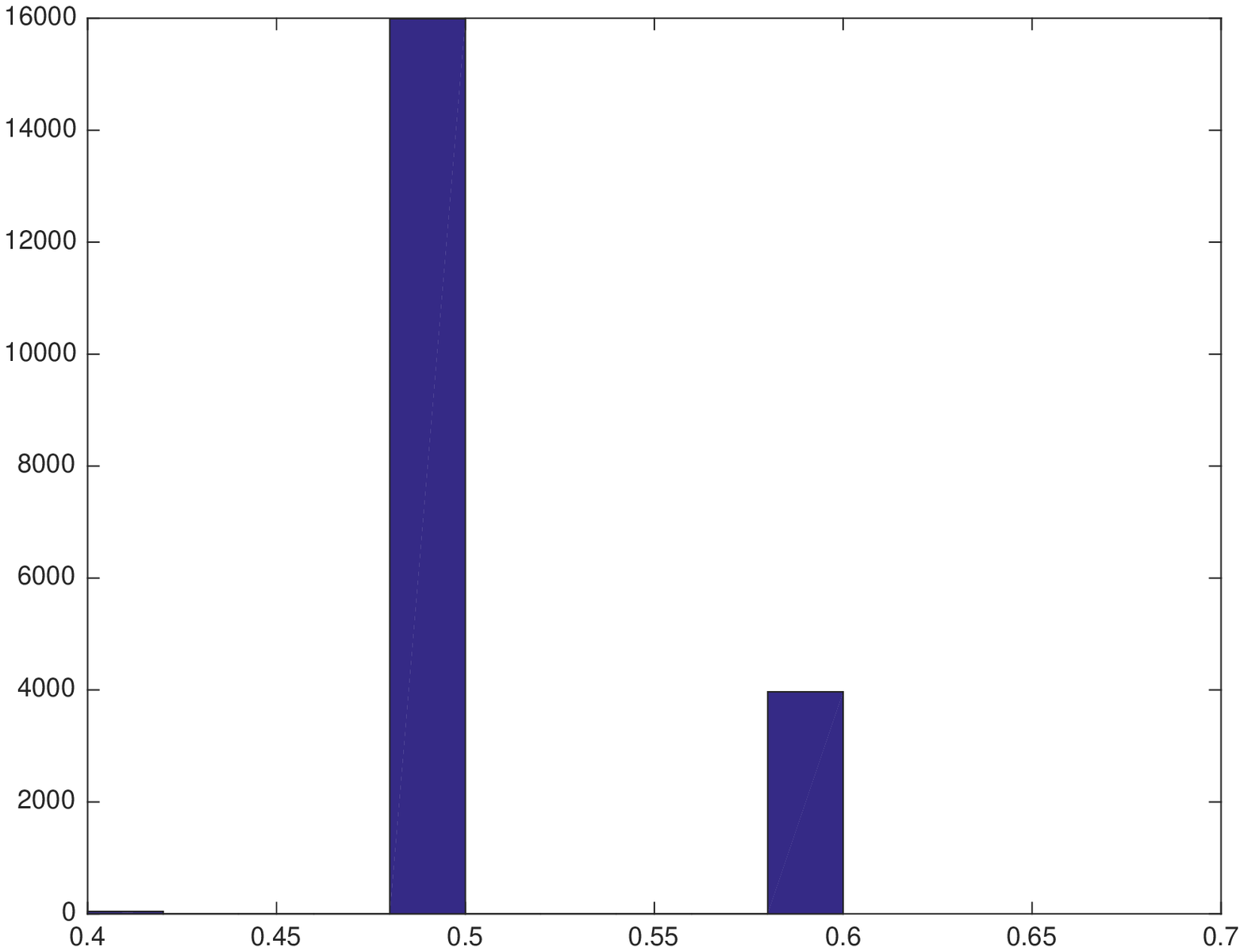}}
  	\subfigure[]
   	{\includegraphics[width=6.75cm]{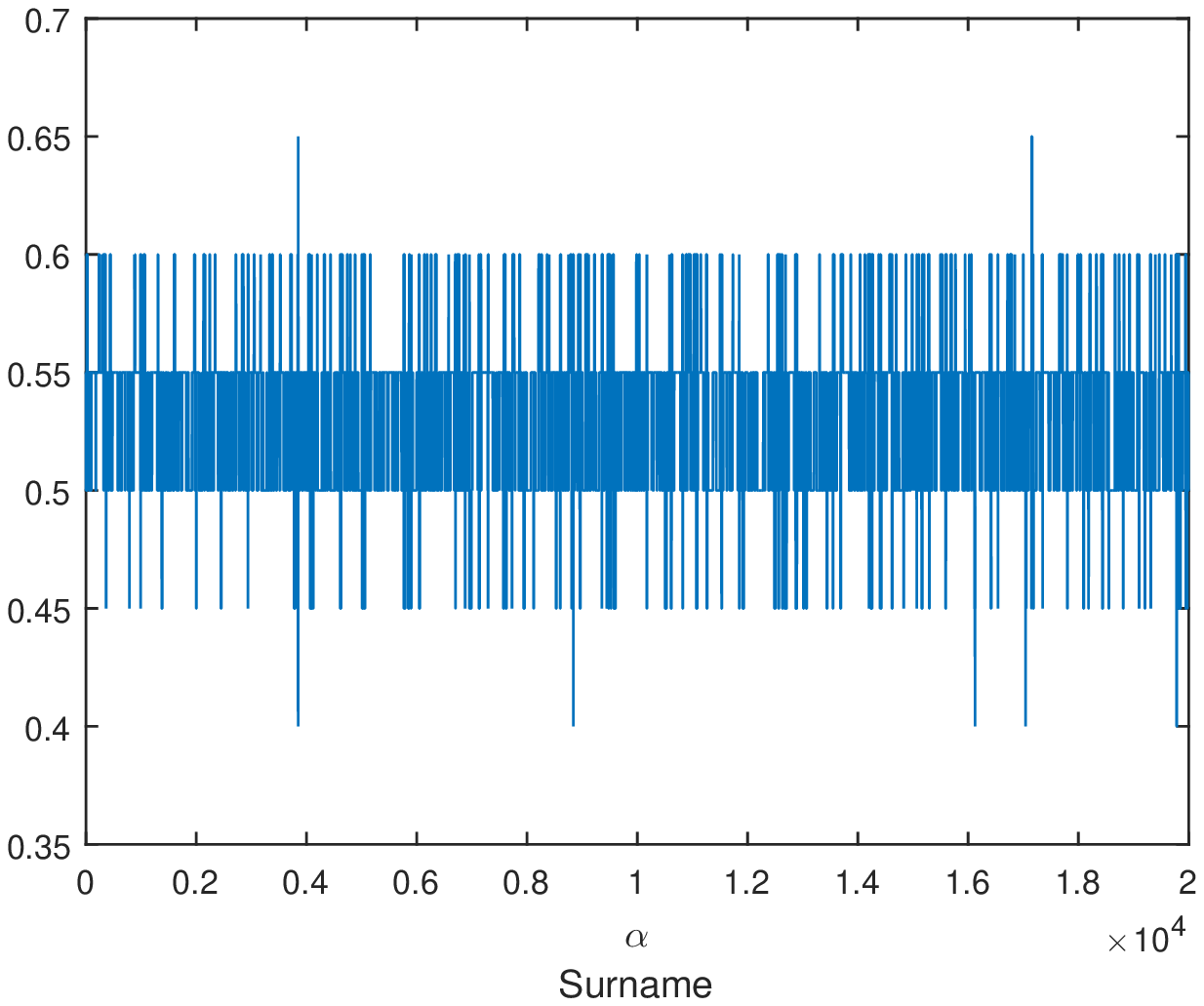}}
  	\subfigure[]
   	{\includegraphics[width=6.75cm]{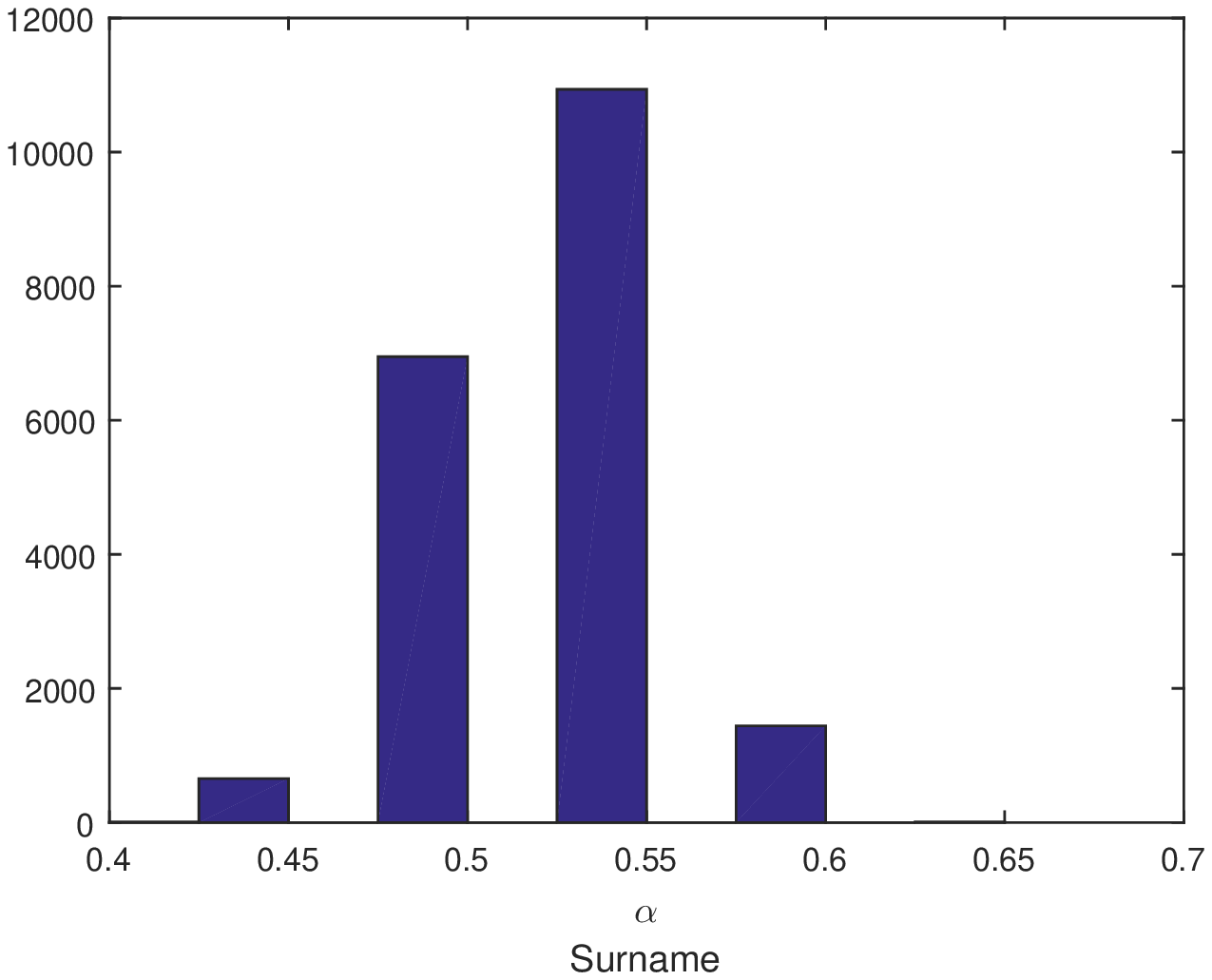}}
\caption{Posterior sample (left) and posterior histogram (right) for the surname data set obtained by applying the Jeffreys prior (top), the loss-based prior with $M=10$ (middle) and the loss-based prior with $M=20$ (bottom).}
\label{Surname}
\end{figure}

\begin{table}
\centering
\begin{tabular}{cccc}
\hline
 Prior & Mean & Median & $95\%$  C. I. \\
\hline
Jeffreys & 0.53 & 0.54 & (0.47, 0.58) \\
Loss-based $(M=10)$ & 0.52 & 0.5 & (0.5, 0.6) \\
Loss-based $(M=20)$ & 0.53 & 0.55 & (0.45, 0.60) \\
\hline
\end{tabular}
\caption{Summary statistics of the posterior distributions for the parameter $\alpha$ of the Census surname analysis.}
\label{TSurn}
\end{table}
The estimated value of $\alpha$, on the basis of the 500 most common surnames in the US (and if we consider the mean) is, roughly, $1/2$. In other words, there are about $50\%$ chances that the next observed surname is not in the list of the 500. Obviously, a larger sample size would yield a smaller posterior mean, as the number of surnames is finite and the more we observe, the harder is to find a ``new'' one.

\subsection{`Superstardom' analysis}
The last example consists in modelling the number of `number one' hits a music artist had in the period 1955--2003 on the Billboard Hot 100 chart. The data, which is displayed in Table \ref{HITS-Data}, has been used by \cite{ChungCox94} and \cite{SpiVoo09} to show an apparent absence of correlation between talent and success in the music industry.

\begin{table}[h!]
\centering
\begin{tabular}{cc|cc}
\hline 
Hits & Observations & Hits & Observations \\ 
\hline 
1 & 119 & 9 & 4 \\ 
%\hline 
2 & 57 & 10 & 2 \\ 
%\hline 
3 & 30 & 11 & 1 \\ 
%\hline 
4 & 13 & 12 & 2 \\ 
%\hline 
5 & 10 & 13 & 1 \\ 
%\hline 
6 & 4 & 14 & 1 \\ 
%\hline 
7 & 1 & 15 & 1 \\ 
%\hline 
8 & 1 & 16 & 1 \\ 
\hline 
\end{tabular}
\caption{Number of `number one' hits per artist from 1955 to 2003.}
\label{HITS-Data}
\end{table}

We have run the Monte Carlo simulation for 25,000 iterations, with a burn in period of 5,000, for each of the considered priors. The posterior samples and histograms are shown in Figure \ref{Hits}, with the correspondinf statistic summaries in Table \ref{THits}.

\begin{figure}[h!]
\centering
	\subfigure[]%Jeffreys]
	{\includegraphics[width=6.75cm]{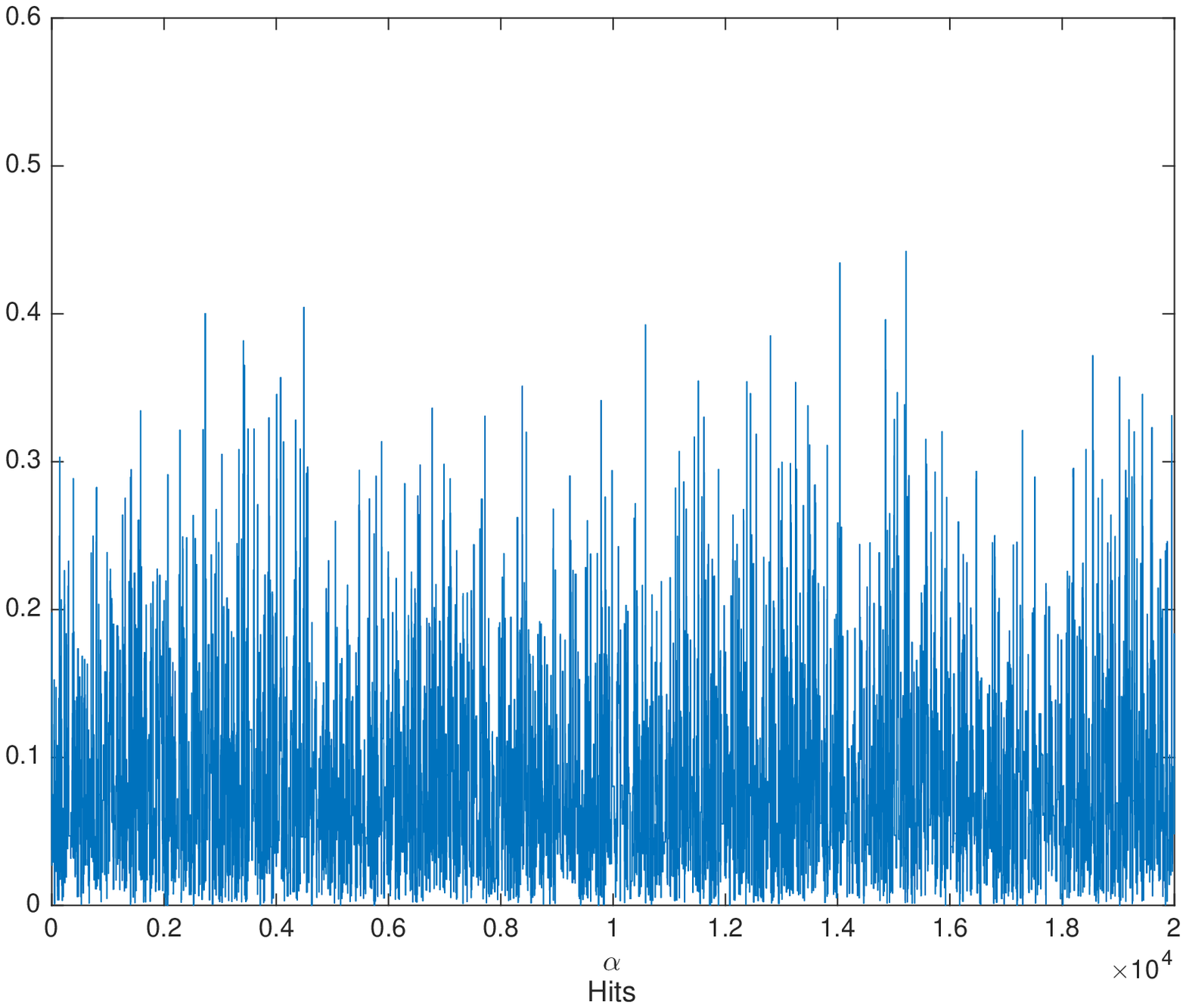}}
	\subfigure[]%Jeffreys]
    {\includegraphics[width=6.75cm]{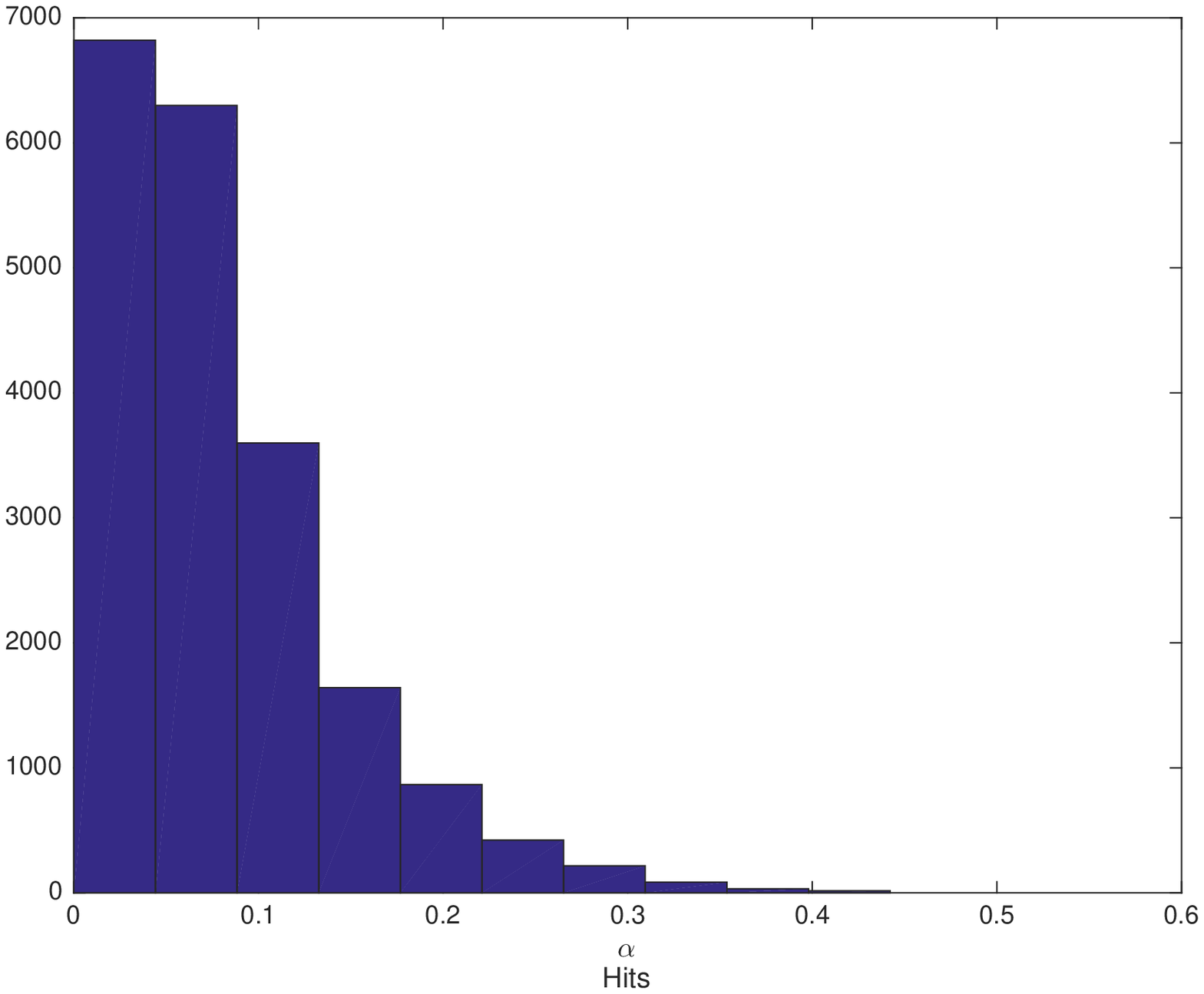}}
    \subfigure[]%loss-prior with $M=10$]
   	{\includegraphics[width=6.75cm]{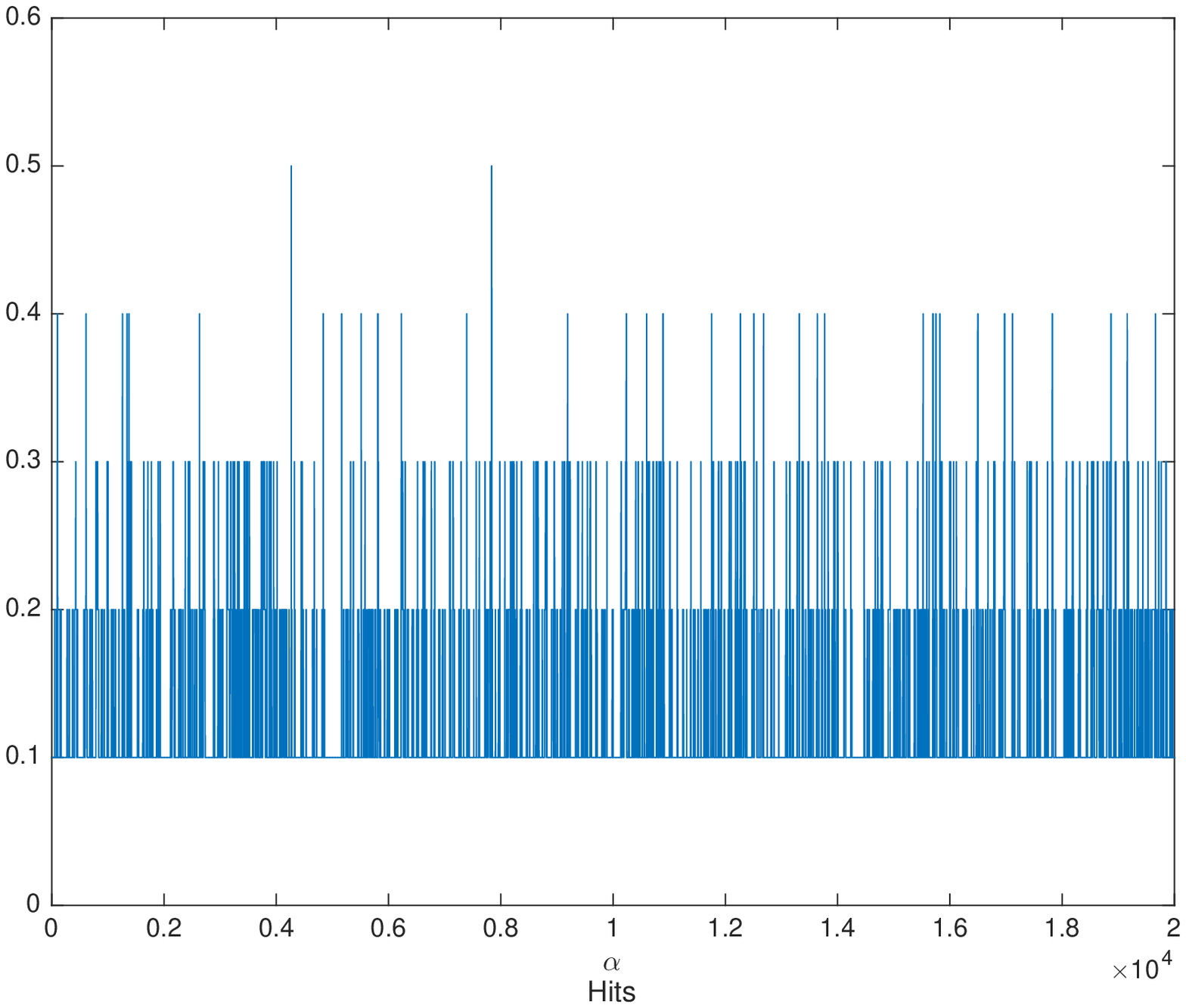}}
   	\subfigure[]%loss-prior with $M=10$]
   	{\includegraphics[width=6.75cm]{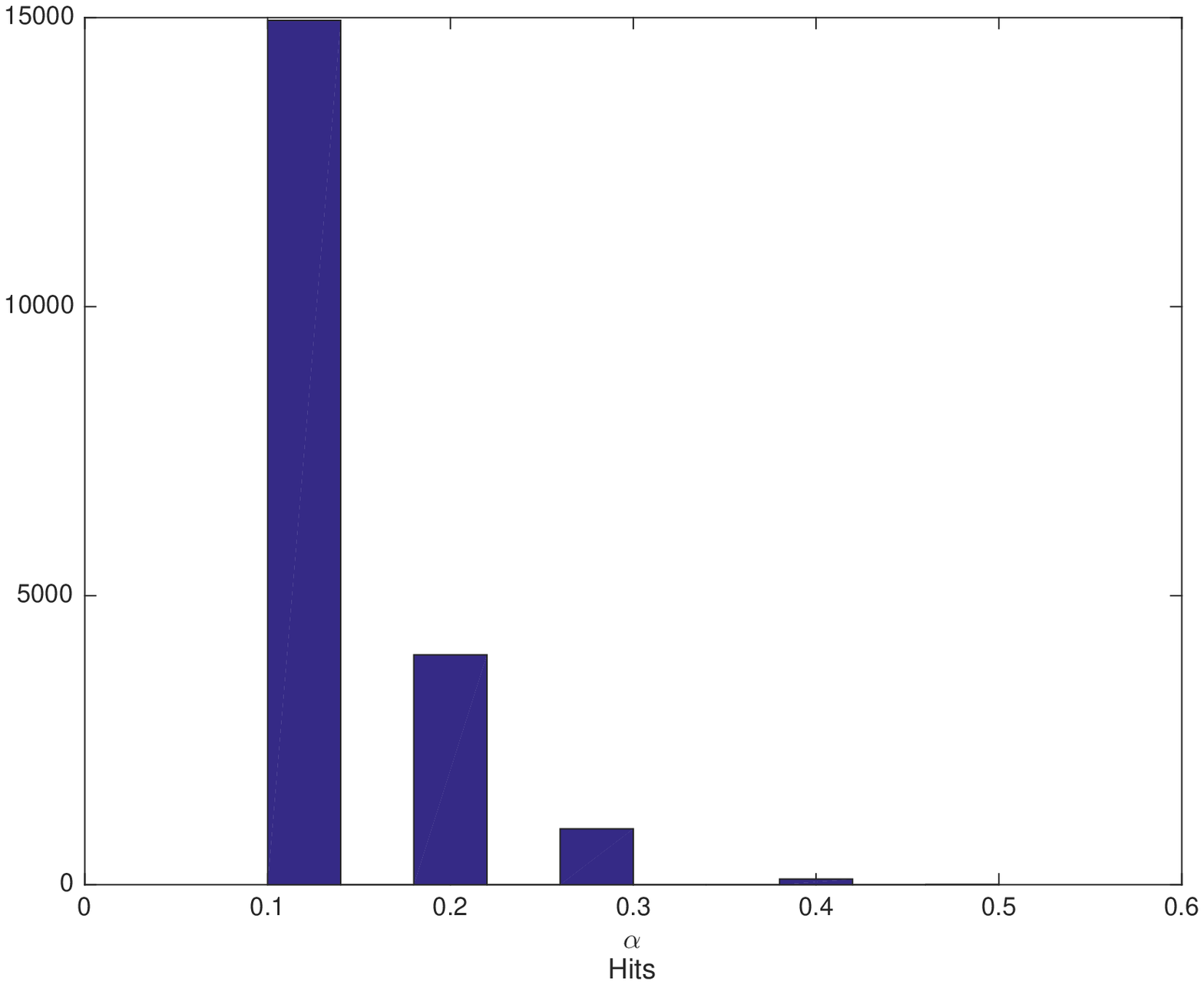}}
  	\subfigure[]%loss-prior with $M=20$]
   	{\includegraphics[width=6.75cm]{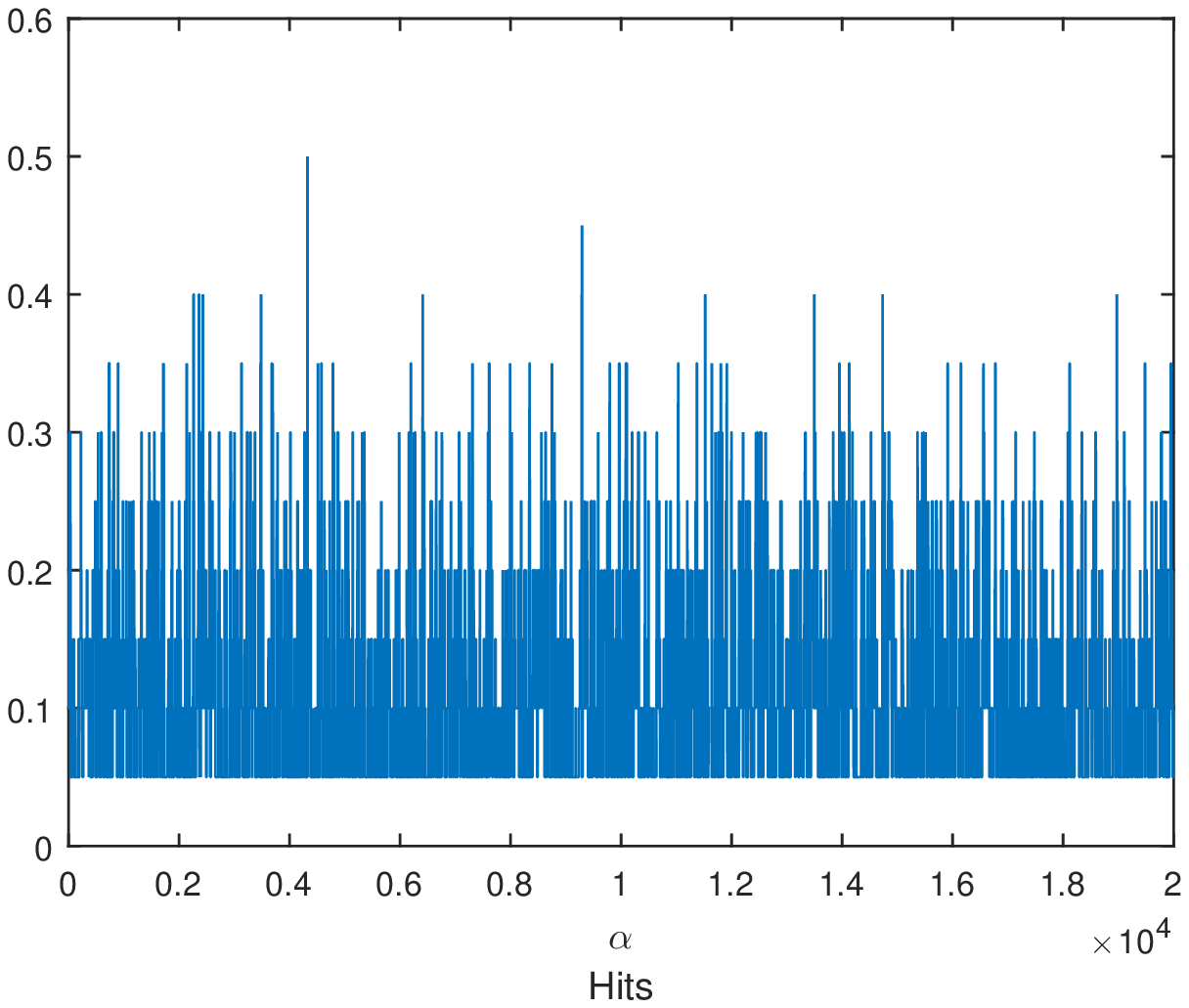}}
    \subfigure[]%loss-prior with $M=20$]
   	{\includegraphics[width=6.75cm]{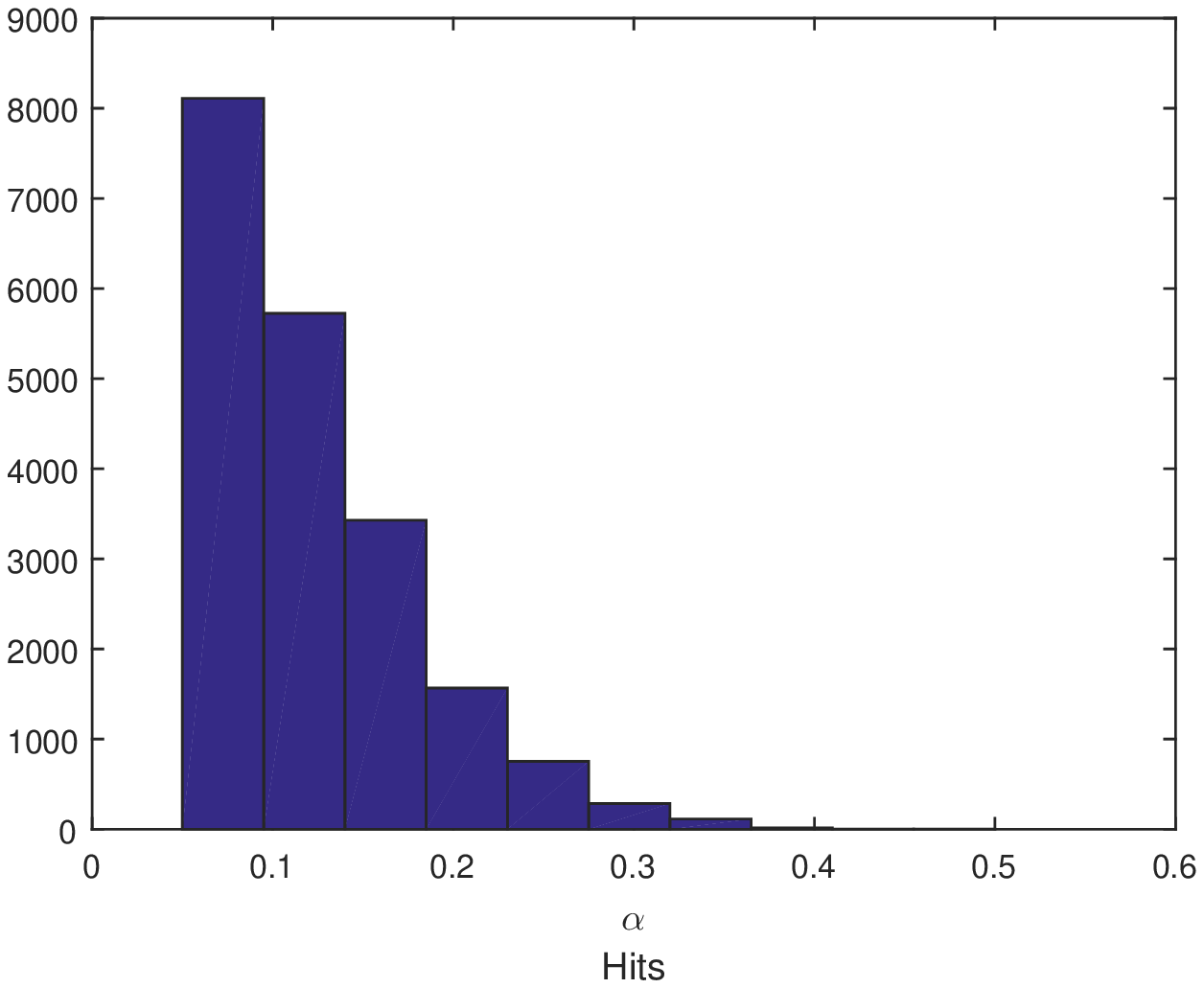}}
\caption{Posterior sample (left) and posterior histogram (right) for the music `number one' hits data set obtained by applying the Jeffreys prior (top), the loss-based prior with $M=10$ (middle) and the loss-based prior with $M=20$ (bottom).}
\label{Hits}
\end{figure}

\begin{table}[h!]
\centering
\begin{tabular}{cccc}
\hline
Prior & Mean & Median & $95\%$  C.I. \\
\hline
Jeffreys & 0.08 & 0.07 & (0.004, 0.24) \\
Loss-based $(M=10)$ & 0.13 & 0.1 & (0.1, 0.3) \\
Loss-based $(M=20)$ & 0.11 & 0.10 & (0.05, 0.25) \\
Loss-based $(M=100)$ & 0.10 & 0.08 & (0.01, 0.29) \\
\hline
\end{tabular}
\caption{Summary statistics of the posterior distribution for the parameter $\alpha$ of the analysis of the music `number one' hits.}
\label{THits}
\end{table}
This example of the music hits allows for some interesting points of dicussion. First, we note that the posterior distributions of for $\alpha$ are skewed; therefore, the posterior median represents a better centrality index than the posterior mean. Second, it is clear that the ``true'' value of $\alpha$ may be close to zero. As such, in order to explore better the parameter space when the loss-based prior is used, a denser discretization is more appropriate. We have then considered $M=100$, resulting the posterior summary statistics in Table \ref{THits}. We note now that the posterior median is similar to the one obtained using the Jeffreys prior. It is therefore recomendable that, when the inference on $\alpha$ indicates values near the parameter space boudaries, the level of discretization to be considered should be relatively dense.

%%%%%%%%%%%%%%%%%%%%%%%%%%%%%%%%%%%%%%%%%%%%%%%%%%%%%%
%% 	 	Conclusion		%%%
%%%%%%%%%%%%%%%%%%%%%%%%%%%%%%%%%%%%%%%%%%%%%%%%%%%%%%
\section{Discussions}
\label{Concl}
It is surprising how, from time to time, the Bayesian literature presents gaps even for problems which appear to be straightforward. The Yule--Simon distribution has undoubtedly many possibilities of application, as the discussed examples and the refereed papers show, and therefore demanded for a satisfactory discussion within the Bayesian framework.

Given the importance that objective Bayesian analysis can have in applications, and not only \citep{Berger06}, we have presented two priors which are suitable in scenarios with minimal prior information. The first prior is the Jeffreys prior which, as it is well known, has the appealing property of being invariant under monotone differentiable transformations of the parameter of interest. The second prior is derived considering the loss in information one would incur if the `wrong' model was selected. Although the latter requires a discretization of the parameter space, we have shown through simulation studies that the performance of the yielded posterior are very similar, both between the Jeffreys and the loss-based prior, and between different structures of the discretized parameter space. This is not surprising as both priors, i.e. the Jeffreys and the loss-based, have a similar behaviour, in the sense that they increase as the parameter $\alpha$ increases.

We have limited our analysis to the case where the shape parameter of the Yule--Simon distribution, $\rho$, is strictly larger than one. Doing so, we allow for a more convenient parametrization of the distribution where the new parameter $\alpha=(\rho-1)/\rho$ has the interpretation of being the probability that the next observation takes a value not observed before.

Besides through a simulation study, we have compared the objective priors by applying them on three data sets: the first related to financial data, the second to surnames in the US and the third one on the number of hits in the music industry. All comparisons allowed to show that the two proposed objective priors lead to similar results, in terms of posterior distributions. For obvious reasons, we have not considered if the choice of the Yule--Simon is the \emph{best} model to represent the data, but limited our analysis to make inference for the unknown parameter $\alpha$.

\section*{Acknowledgements}
Fabrizio Leisen was supported by the European Community's Seventh Framework Programme [FP7/2007-2013] under grant agreement no: 630677.
\vspace*{-8pt}

%%%%%%%%%%%%%%%%%%%%%%%%%%%%%%%%%%%%%%%%%%%%%%%%%%%%%%
%%          Bibliography 		%%%
%%%%%%%%%%%%%%%%%%%%%%%%%%%%%%%%%%%%%%%%%%%%%%%%%%%%%%
%\bibliographystyle{apalike}
\bibliographystyle{biom}
\bibliography{YuSimoBiblio}

%%%%%%%%%%%%%%%%%%%%%%%%%%%%%%%%%%%%%%%%%%%%%%%%%%%%%%
%%       	  Appendix			%%%
%%%%%%%%%%%%%%%%%%%%%%%%%%%%%%%%%%%%%%%%%%%%%%%%%%%%%%
\clearpage

\renewcommand{\thesection}{A}
\renewcommand{\theequation}{A.\arabic{equation}}
\renewcommand{\thefigure}{A.\arabic{figure}}
\renewcommand{\thetable}{A.\arabic{table}}
\setcounter{table}{0}
\setcounter{figure}{0}
\setcounter{equation}{0}

\section{Appendix}
\label{Proof}
\begin{proof}[Proof of Theorem \ref{Jef1}]

First of all, we note that
\begin{align}
&\frac{\partial ^{2}\log(f(k;\alpha))}{\partial \alpha^{2}}
= \frac{1}{(1-\alpha)^{2}}+\frac{2}{(1-\alpha)^{3}}\biggl[\psi^{(0)}\biggl(\frac{1}{1-\alpha}+1\biggr)\notag \\ &\hspace{2cm}-\psi^{(0)}\biggl(\frac{1}{1-\alpha}+k+1\biggr)\biggr] + 
 \frac{1}{(1-\alpha)^{4}}\biggl[\psi^{(1)}\biggl(\frac{1}{1-\alpha}+1\biggr)\notag \\&\hspace{7cm}-\psi^{(1)}\biggl(\frac{1}{1-\alpha}+k+1\biggr)\biggr] \notag\label{2nd1}
\end{align}
where $\psi^{(i)}$ is the polygamma function:
\begin{equation}
\psi^{(i)}(x)=\frac{\partial^{i+1}}{\partial x^{i+1}} \log(\Gamma(x))=(-1)^{i+1}i! \sum_{k=0}^{\infty} \frac{1}{(x+k)^{i+1}} \quad i=1,2,\dots\notag\label{dig}.
\end{equation}
It's easy to see that 
$$\psi^{(0)}\biggl(\frac{1}{1-\alpha}+1\biggr)-\psi^{(0)}\biggl(\frac{1}{1-\alpha}+k+1\biggr)=\sum_{j=1}^{k}\frac{1}{\frac{1}{1-\alpha}+j}$$
and
$$\psi^{(1)}\biggl(\frac{1}{1-\alpha}+1\biggr)-\psi^{(1)}\biggl(\frac{1}{1-\alpha}+k+1\biggr)=\sum_{j=0}^{k-1}\frac{1}{(\frac{1}{1-\alpha}+1+j)^{2}}$$
Therefore, we have that the Fisher information is:
\begin{align}
\mathcal{I}(\alpha)& =-\frac{1}{(1-\alpha)^{2}}+\frac{2}{(1-\alpha)^{3}}\mathbb{E}_{\alpha}\biggl[\sum_{j=1}^{k}\frac{1}{\frac{1}{1-\alpha}+j}\biggr] \notag \\
& \hspace{2cm}- \frac{1}{(1-\alpha)^{4}}\mathbb{E}_{\alpha}\biggl[\sum_{j=0}^{k-1}\frac{1}{(\frac{1}{1-\alpha}+1+j)^{2}} \biggr]  \label{Fisher}
\end{align}
In order to compute the Jeffreys prior, we need compute the two expected value of equation \ref{Fisher} separately. 

\begin{align}
\mathbb{E}_{\alpha}\biggl[\sum_{j=1}^{k}\frac{1}{(\frac{1}{1-\alpha}+j)} \biggr]&= \sum_{k=1}^{\infty} \sum_{j=1}^{k} \frac{1}{(\frac{1}{1-\alpha}+j)} \frac{1}{1-\alpha} B\biggl(k,\frac{1}{1-\alpha}+1\biggr) \notag \\
&=\sum_{j=1}^{\infty} \frac{1}{(\frac{1}{1-\alpha}+j)}  \sum_{k=j}^{\infty} \frac{1}{1-\alpha} B\biggl(k,\frac{1}{1-\alpha}+1\biggr) \label{pippo1}
\end{align}
The second summation with respect to $k$ in equation \ref{pippo1} can be rewritten as:
\begin{align}
\sum_{k=j}^{\infty}\frac{1}{1-\alpha} B\biggl(k,\frac{1}{1-\alpha}+1\biggr)&= \sum_{k=j}^{\infty}\frac{1}{1-\alpha} \int_{0}^{1} x^{k-1}(1-x)^{\frac{1}{1-\alpha}}dx \notag \\
&= \sum_{l=0}^{\infty} \frac{1}{1-\alpha} \int_{0}^{1}x^{l}x^{j-1}(1-x)^{\frac{1}{1-\alpha}}dx \notag \\
&=\frac{1}{1-\alpha} \int_{0}^{1}x^{j-1}(1-x)^{\frac{1}{1-\alpha}-1} dx\notag \\
%&= \frac{1}{1-\alpha} B\biggl(j,\frac{1}{1-\alpha}\biggr)\notag\\
%&=\frac{1}{1-\alpha} \frac{\Gamma(j)\Gamma(\frac{1}{1-\alpha})}{\Gamma(\frac{1}{1-\alpha}+j)}\notag \\
&=\frac{\Gamma(j)\Gamma(\frac{1}{1-\alpha}+1)}{\Gamma(\frac{1}{1-\alpha}+j)} \label{SumInterna}
\end{align}
Finally, we have that :
\begin{align}
\mathbb{E}_{\alpha}\biggl[\sum_{j=1}^{k}\frac{1}{(\frac{1}{1-\alpha}+j)} \biggr]&=\Gamma(\frac{1}{1-\alpha}+1)  \sum_{j=1}^{\infty} \frac{\Gamma(j)}{\Gamma(\frac{1}{1-\alpha}+j) (\frac{1}{1-\alpha}+j)} \notag \\
&=\Gamma\biggl(\frac{1}{1-\alpha}+1\biggr)  \sum_{j=1}^{\infty} \frac{\Gamma(j)}{\Gamma(\frac{1}{1-\alpha}+j+1)} \notag \\
&=\sum_{j=1}^{\infty}\biggl(\frac{1-\alpha}{1-\alpha}\biggr) \frac{\Gamma(\frac{1}{1-\alpha}+1)\Gamma(j)}{\Gamma(\frac{1}{1-\alpha}+j+1)}\notag\\
&=(1-\alpha) \sum_{j=1}^{\infty}\frac{1}{1-\alpha} \frac{\Gamma(\frac{1}{1-\alpha}+1)\Gamma(j)}{\Gamma(\frac{1}{1-\alpha}+j+1)}= \notag \\
&=(1-\alpha) \label{Fin1st}
\end{align}
where the summation $\sum_{j=1}^{\infty} (\frac{1}{1-\alpha})B(j,\frac{1}{1-\alpha}+1)=1$, since we are summing over all the possible values of the probability function of the Yule-Simon distribution.

As we have done with the first expected value of \ref{Fisher}, now we compute the second expected value of equation \ref{Fisher}:
\begin{align}
\mathbb{E}_{\alpha}\biggl[\sum_{j=0}^{k-1}\frac{1}{(\frac{1}{1-\alpha}+1+j)^{2}} \biggr]&= \sum_{k=1}^{\infty} \sum_{j=0}^{k-1} \frac{1}{(1+\frac{1}{1-\alpha}+j)^{2}} \frac{1}{1-\alpha} \frac{\Gamma(k)\Gamma(\frac{1}{1-\alpha}+1)}{\Gamma(\frac{1}{1-\alpha}+1+k)} \notag \\
&=\sum_{k=1}^{\infty} \sum_{j=1}^{k} \frac{1}{(\frac{1}{1-\alpha}+j)^{2}} \biggl(\frac{1}{1-\alpha} \biggr)\frac{\Gamma(k)\Gamma(\frac{1}{1-\alpha}+1)}{\Gamma(\frac{1}{1-\alpha}+1+k)} \notag \\
&=\sum_{j=1}^{\infty} \frac{1}{(\frac{1}{1-\alpha}+j)^{2}}  \sum_{k=j}^{\infty} \frac{1}{1-\alpha} B\biggl(k,\frac{1}{1-\alpha}+1\biggr)\label{pippo}\notag\\
&=\sum_{j=1}^{\infty} \frac{1}{(\frac{1}{1-\alpha}+j)^{2}}\frac{\Gamma(j)\Gamma(\frac{1}{1-\alpha}+1)}{\Gamma(\frac{1}{1-\alpha}+j)}
\end{align}
where the last equality follows from \eqref{SumInterna}.Finally we obtain the following form:
\begin{align}
\mathbb{E}_{\alpha}\biggl[\sum_{j=0}^{k-1}\frac{1}{(\frac{1}{1-\alpha}+1+j)^{2}} \biggr]= \sum_{j=1}^{\infty} \frac{B(j,\frac{1}{1-\alpha}+1)}{\frac{1}{1-\alpha}+j}\label{Hyper}
%&>\Gamma(\rho+1) \sum_{j=1}^{\infty} \frac{1}{j(j+1)\cdots(j+\rho)(j+1+\rho)} = \notag \\
%&>\Gamma\biggl(\frac{1}{1-\alpha}+1\biggr)  \sum_{j=1}^{\infty} \frac{\Gamma(j)}{\Gamma(\frac{1}{1-\alpha}+j+2)}= \notag \\
%&=\frac{\Gamma(\frac{1}{1-\alpha}+1)}{\Gamma(\frac{1}{1-\alpha}+2)(\frac{1}{1-\alpha}+1)}  \sum_{j=1}^{\infty} \frac{(\frac{1}{1-\alpha}+1)\Gamma(\frac{1}{1-\alpha}+2)\Gamma(j)}{\Gamma(\frac{1}{1-\alpha}+j+2)}= \notag \\
%&=\frac{\Gamma(\frac{1}{1-\alpha}+1)}{\Gamma(\frac{1}{1-\alpha}+2)(\frac{1}{1-\alpha}+1)} \sum_{j=1}^{\infty} \biggl(\frac{1}{1-\alpha}+1\biggr)B\biggl(j,\frac{1}{1-\alpha}+2\biggr) \notag \\
%&= \frac{1}{(\frac{1}{1-\alpha}+1)^{2}} =\frac{1}{(\frac{2-\alpha}{1-\alpha})^{2}}=\frac{(1-\alpha)^{2}}{(2-\alpha)^{2}}\label{Fin2nd}
\end{align}
%where the summation $\sum_{j=1}^{\infty} (\frac{1}{1-\alpha}+1)B(j,\frac{1}{1-\alpha}+2)=1$, since it is a Yule-Simon distribution.

%\begin{rem}
The equation \ref{Hyper} can be written in a more simple way as a function of an Hypergeometric function. 
\begin{align}
\sum_{j=1}^{\infty} \frac{B(j,\frac{1}{1-\alpha}+1)}{\frac{1}{1-\alpha}+j}&=\sum_{j=1}^{\infty} \frac{1}{\frac{1}{1-\alpha}+j}\int_{0}^{1}x^{j-1}(1-x)^{\frac{1}{1-\alpha}} dx\notag \\
&=\int_{0}^{1} (1-x)^{\frac{1}{1-\alpha}}\sum_{j=1}^{\infty} \frac{1}{\frac{1}{1-\alpha}+j}x^{j-1} dx \label{HH}
\end{align}
Looking at the summation we have:
\begin{align}
\sum_{j=1}^{\infty}& \frac{1}{\frac{1}{1-\alpha}+j}x^{j-1} dx =\sum_{l=0}^{\infty} \frac{x^{l}}{\frac{1}{1-\alpha}+l+1} \label{Poch}
\end{align}
But the denominator can be written as a ratio of Pochhammer representations:
\begin{equation}
\frac{1}{\frac{1}{1-\alpha}+l+1}=\frac{\bigl(\frac{1}{1-\alpha}+1\bigr)_{l}}{\bigl(\frac{1}{1-\alpha}+2\bigr)_{l}} \frac{1}{\frac{1}{1-\alpha}+1} \notag
\end{equation}
Hence the equation \ref{Poch} is written as:
\begin{align}
\sum_{l=0}^{\infty}\frac{x^{l}}{\frac{1}{1-\alpha}+l+1}& =\frac{1}{\frac{1}{1-\alpha}+1}\sum_{l=0}^{\infty} \frac{\bigl(\frac{1}{1-\alpha}+1\bigr)_{l}}{\bigl(\frac{1}{1-\alpha}+2\bigr)_{l}} x^{l}\frac{(1)_{l}}{l!}= \notag \\
&= \frac{1}{\frac{1}{1-\alpha}+1} \, _{2}F_{1}\biggl(1,\frac{1}{1-\alpha}+1,\frac{1}{1-\alpha}+2,x\biggr) \notag
\end{align}
where $_{2}F_{1}(\alpha,\beta,\gamma,x)$ is the hypergeometric function. So we have that equation \ref{HH} can be rewritten as:
\begin{align}
\int_{0}^{1} &\frac{(1-x)^{\frac{1}{1-\alpha}}}{\frac{1}{1-\alpha}+1} \, _{2}F_{1}\biggl(1,\frac{1}{1-\alpha}+1,\frac{1}{1-\alpha}+2,x\biggr)= \notag \\
%&= \frac{1}{\frac{1}{1-\alpha}+1}B\biggl(1,\frac{1}{1-\alpha}+1\biggr) \, _{3}F_{2}\biggl(1,\frac{1}{1-\alpha}+1,1; \frac{1}{1-\alpha}+2,\frac{1}{1-\alpha}+2;1\biggr) \notag \\
&=\frac{(1-\alpha)^{2}}{(2-\alpha)^{2}}\, _{3}F_{2}\biggl(1,\frac{1}{1-\alpha}+1,1; \frac{1}{1-\alpha}+2,\frac{1}{1-\alpha}+2;1\biggr) \notag
\end{align}
where the last equality follows from 7.512.5 of \cite{GR07}. Summing up,
\begin{align}
\mathcal{I}(\alpha)= -\frac{1}{(1-\alpha)^{2}}&+\frac{2}{(1-\alpha)^{3}}(1-\alpha) +\notag \\
- \frac{1}{(1-\alpha)^{4}}\frac{(1-\alpha)^{2}}{(2-\alpha)^{2}}&\, _{3}F_{2}\biggl(1,\frac{1}{1-\alpha}+1,1; \frac{1}{1-\alpha}+2,\frac{1}{1-\alpha}+2;1\biggr) \notag 
\end{align}
and this concludes the proof.
\end{proof}

\end{document}